\documentclass[11pt]{article}
\usepackage[T1]{fontenc}
\usepackage{amsfonts,amsmath,amsthm,amssymb,dsfont,mathtools}
\usepackage[margin=1in]{geometry}
\usepackage{fullpage}
\usepackage{enumerate}
\usepackage[linesnumbered,boxed,ruled,vlined]{algorithm2e}
\usepackage{thm-restate}
\usepackage[colorlinks,citecolor=blue,linkcolor=blue,urlcolor=red,pagebackref]{hyperref}
\usepackage[capitalise]{cleveref}
\usepackage{url}
\usepackage{graphicx}
\usepackage{mathdots}
\usepackage{paralist}
\usepackage[normalem]{ulem}
\usepackage{xspace}
\xspaceaddexceptions{]\}}
\usepackage{comment}
\usepackage{tabu}
\usepackage{framed}
\usepackage{float,wrapfig}
\usepackage{subfigure}
\usepackage{tikz}
\usepackage[usenames,dvipsnames]{pstricks}

\usepackage[mathlines]{lineno}
\usepackage{etoolbox}
\newcommand*\linenomathpatch[1]{%
  \cspreto{#1}{\linenomath}%
  \cspreto{#1*}{\linenomath}%
  \csappto{end#1}{\endlinenomath}%
  \csappto{end#1*}{\endlinenomath}%
}
\linenomathpatch{equation}
\linenomathpatch{gather}
\linenomathpatch{multline}
\linenomathpatch{align}
\linenomathpatch{alignat}
\linenomathpatch{flalign}

\theoremstyle{plain}
\newtheorem{theorem}{Theorem}[section]
\newtheorem{lemma}[theorem]{Lemma}

\newtheorem{corollary}[theorem]{Corollary}
\newtheorem{claim}[theorem]{Claim}

\newtheorem{prob}{Problem}

\theoremstyle{definition}
\newtheorem{definition}[theorem]{Definition}

\def\ShowAuthNotes{1}
\ifnum\ShowAuthNotes=1
\newcommand{\authnote}[2]{\ \\ \textcolor{red}{\parbox{0.9\linewidth}{[{\footnotesize {\bf #1:} { {#2}}}]}}\newline}
\else
\newcommand{\authnote}[2]{}
\fi

\newcommand{\eps}{\varepsilon}
\renewcommand{\Pr}{\operatorname*{\mathbf{Pr}}}

\newcommand{\poly}{\operatorname{\mathrm{poly}}}
\newcommand{\polylog}{\poly\log}

\newcommand{\R}{\mathbb{R}}

\newcommand{\Z}{\mathbb{Z}}

\renewcommand{\tilde}{\widetilde}
\newcommand{\bod}{\boldsymbol}
\renewcommand{\vec}{\bod}

\newcommand{\caH}{\mathcal{H}}
\newcommand{\caI}{\mathcal{I}}

\newcommand{\caS}{\mathcal{S}}

\newcommand{\caW}{\mathcal{W}}

\newcommand{\ww}{w_{\max}}
\newcommand{\pp}{p_{\max}}

\newcommand{\dd}{\mathinner{.\,.\allowbreak}}
	
\newcommand{\supp}{\operatorname{\mathrm{supp}}}
\newcommand{\wts}{\operatorname{\mathrm{weights}}}

\newcommand{\erdos}{Erd\H{o}s\xspace}
\newcommand{\sarkozy}{S\'{a}rk\"{o}zy\xspace}

\title{Solving Knapsack with Small Items via $\ell_0$-Proximity}
\author{ Ce Jin\thanks{cejin@mit.edu. Supported by NSF CCF-2129139, CCF-2127597, and a Siebel Scholarship.}\\MIT}
\date{\vspace{-1cm}}

\begin{document}

	\setcounter{page}{0} \clearpage
	\maketitle
	\thispagestyle{empty}
	\begin{abstract}
		We study pseudo-polynomial time algorithms for the fundamental \emph{0-1 Knapsack} problem.
		Advances in fine-grained complexity have shown conditional optimality for the textbook $O(nt)$-time algorithm (Bellman 1957), where $n$ is the number of items and $t$ is the knapsack capacity. 
		To cope with this hardness, recent interest has focused on the dependency on \emph{maximum item size} $w_{\max}$, aiming for faster algorithms when the items are small.

		In terms of $n$ and $w_{\max}$, previous algorithms for 0-1 Knapsack have cubic time complexities: $O(n^2w_{\max})$ (Bellman 1957), $O(nw_{\max}^2)$ (Kellerer and Pferschy 2004), 
		and $O(n + w_{\max}^3)$ (Polak, Rohwedder, and W\k{e}grzycki 2021). 
		On the other hand, fine-grained complexity only rules out $O((n+w_{\max})^{2-\delta})$ running time, and it is an important question in this area whether $\tilde O(n+w_{\max}^2)$ time is achievable.
        Our main result makes significant progress towards solving this question:
		\begin{itemize}
			\item The 0-1 Knapsack problem has a deterministic algorithm in $\tilde O(n + w_{\max}^{2.5})$ time.  
		\end{itemize}
		Our techniques also apply to the easier  \emph{Subset Sum} problem:
		\begin{itemize}
			\item The Subset Sum problem has a randomized algorithm in $\tilde O(n + w_{\max}^{1.5})$ time.
		\end{itemize}
		 This improves (and simplifies) the previous $\tilde O(n + w_{\max}^{5/3})$-time algorithm by Polak, Rohwedder, and W\k{e}grzycki (2021) (based on Galil and Margalit (1991), and Bringmann and Wellnitz (2021)).

Similar to recent works on Knapsack (and integer programs in general), our algorithms also utilize the \emph{proximity} between optimal integral solutions and fractional solutions.
		Our new ideas are as follows:
		\begin{enumerate}[(1)]
			\item  
		 Previous works used an $O(w_{\max})$ proximity bound in the $\ell_1$-norm. As our main conceptual contribution,  we use an additive-combinatorial theorem by Erd\H{o}s and S\'{a}rk\"{o}zy (1990) to derive an  $\ell_0$-proximity bound of $\tilde O(\sqrt{w_{\max}})$.
		\item Then, the main technical component of our Knapsack result is a dynamic programming algorithm that exploits both $\ell_0$- and $\ell_1$-proximity.
		It is based on a vast extension of the ``witness propagation'' method, originally designed by Deng, Mao, and Zhong (2023) for the easier \emph{unbounded} setting only.
	Other ingredients include SMAWK algorithm on tall matrices, (derandomized) Bringmann's color-coding, and a novel pruning method.
		\end{enumerate}

	\end{abstract}
	\newpage

\maketitle

\section{Introduction}

In the \emph{0-1 Knapsack} problem, we are given a knapsack capacity $t\in \Z^+$ and $n$ items $(w_1,p_1)$, $\dots$, $(w_n,p_n)$, where $w_i,p_i \in \Z^+$ denote the \emph{weight} and \emph{profit} of the $i$-th item, and we want to select a subset $S\subseteq [n]$ of items satisfying the capacity constraint $\sum_{i\in S}w_i \le t$, while maximizing the total profit $\sum_{i\in S}p_i$.
The \emph{Subset Sum} problem is the special case of Knapsack where $w_i=p_i$ for all items~$i$.

Knapsack and Subset Sum are fundamental problems in computer science.  They are among Karp's 21 NP-complete problems \cite{karp1972reducibility}, and the fastest known algorithms solving them run in $O(2^{n/2})$ time \cite{horowitz1974computing,schroeppel1981t}.
However, when the input integers are small, it is more preferable to use \emph{pseudopolynomial time} algorithms that have polynomial time dependence on both $n$ and the input integers.
Our work focuses on this pseudopolynomial regime. A well-known example of  pseudopolynomial algorithms is the textbook $O(nt)$-time Dynamic Programming (DP) algorithm for Knapsack and Subset Sum, given by Bellmann \cite{Bellman:1957} in 1957. 
Finding faster pseudopolynomial algorithms for these problems became an important topic in combinatorial optimization and operation research;
 see the book of Kellerer, Pferschy, and Pisinger \cite{DBLP:books/daglib/0010031} for a nice summary of the results known by the beginning of this century.
In the last few years, research on Subset Sum and Knapsack has been revived by recent developments in
fine-grained complexity (e.g, \cite{DBLP:journals/talg/CyganMWW19,DBLP:conf/icalp/KunnemannPS17,DBLP:journals/talg/KoiliarisX19,DBLP:conf/soda/Bringmann17,DBLP:conf/stoc/BateniHSS18,DBLP:journals/talg/AbboudBHS22}) and integer programming (e.g., \cite{DBLP:journals/talg/EisenbrandW20,icalp21}), and the central question is to understand the best possible time complexities for solving these problems.

\cite{DBLP:journals/talg/CyganMWW19} and \cite{DBLP:conf/icalp/KunnemannPS17} showed that the $O(nt)$ time complexity for 0-1 Knapsack is essentially optimal
 (in the regime of $t=\Theta(n)$)
under the $(\min, +)$-convolution hypothesis.
For the easier Subset Sum problem, Bringmann \cite{DBLP:conf/soda/Bringmann17} gave an $\tilde O(n+t)$ time algorithm, and \cite{DBLP:journals/talg/AbboudBHS22} showed that
this linear dependency on $t$ is essentially optimal under Strong ETH.
Hence, it seems not hopeful to obtain significant further improvements to these algorithms.

To cope with these hardness results,
recent interest has focused on parameterizing the running time in terms of $n$ and the \emph{maximum item weight} $\ww$, instead of the knapsack capacity $t$.
This would be especially useful when the item weights are much smaller than the capacity, and results along this line would offer us a more fine-grained understanding of the Knapsack and Subset Sum problems. %
This parameterization  is also natural from the perspective of integer linear programming (e.g., \cite{DBLP:journals/talg/EisenbrandW20}): when formulating Knapsack as an integer linear program, the maximum item weight $\ww$ corresponds to the standard parameter $\Delta$, maximum absolute value in the input matrix.

However, despite extensive research along these lines, our understanding about the dependence on $\ww$ is still incomplete. Known fine-grained lower bounds only ruled out $(n+\ww)^{2-\delta}$ algorithms for Knapsack 
\cite{DBLP:journals/talg/CyganMWW19,DBLP:conf/icalp/KunnemannPS17}
and $2^{o(n)}\ww^{1-\delta}$ algorithms for Subset Sum \cite{DBLP:journals/talg/AbboudBHS22} (for $\delta >0$). 
In comparison, directly plugging the trivial bound $t=O(n\ww)$ into the aforementioned algorithms would only give $\tilde O(n\ww)$ time for Subset Sum and $O(n^2\ww)$ time for Knapsack. A lot of recent works have obtained improved bounds (to be discussed in the following), but they are still far from these lower bounds. The following questions have been asked by 
  \cite{icalp21,DBLP:conf/icalp/BringmannC22}, and by \cite{DBLP:conf/soda/AxiotisBJTW19,DBLP:journals/jcss/AbboudBHS22,DBLP:conf/soda/BringmannW21,icalp21,DBLP:conf/icalp/BringmannC22} respectively:  
  \begin{center}
\textit{Question 1: Can 0-1 Knapsack be solved in $\tilde O(n+\ww^2)$ time? }
	
\textit{Question 2: Can Subset Sum be solved in $\tilde O(n+\ww)$ time?}
\end{center}

Now we summarize known algorithmic results around these questions.
\paragraph*{Knapsack.}
Compared to the baseline $O(n^2 \ww)$-time algorithm via Bellman's dynamic programming, several results obtained the bound $\tilde O(n\ww^2)$ (which gives an improvement for small $\ww$)  via various methods \cite{DBLP:journals/talg/EisenbrandW20,DBLP:conf/stoc/BateniHSS18,
DBLP:conf/icalp/AxiotisT19,DBLP:journals/jco/KellererP04}.
Finally, Polak, Rohwedder, and W\k{e}grzycki
\cite{icalp21} carefully combined the \emph{proximity technique} of
Eisenbrand and Weismantel \cite{DBLP:journals/talg/EisenbrandW20}
from integer programming with
the concave $(\max,+)$ convolution algorithm (\cite{DBLP:journals/jco/KellererP04} or  \cite{smawk}), and obtained an $O(n+\ww^3)$ algorithm for Knapsack. The algorithms of \cite{icalp21} and \cite{DBLP:journals/talg/EisenbrandW20} also apply to a generalized problem called Bounded Knapsack (where each item has a given multiplicity not necessarily $1$).
Comparing with the $(n+\ww)^{2-o(1)}$ fine-grained lower bound, we notice that all these algorithms have cubic dependence on $(n+\ww)$.

Analogous results also exist if we parameterize by maximum item profit $\pp$ instead,  or by both $\pp$ and $\ww$ simultaneously; see the further related works section.

\paragraph*{Subset Sum.}
An early result for Subset Sum in terms of $n$ and $\ww$ is Pisinger's deterministic $O(n\ww)$-time algorithm for Subset Sum \cite{DBLP:journals/jal/Pisinger99}. This is not completely subsumed by Bringmann's $\tilde O(n+t)\le \tilde O(n\ww)$ time algorithm \cite{DBLP:conf/soda/Bringmann17}, due to the extra log factors and randomization in the latter result.

More recently, Polak, Rohwedder, and W\k{e}grzycki \cite{icalp21} observed that an $\tilde O(n+\ww^2)$ time algorithm directly follows from combining their proximity technique with Bringmann's $\tilde O(n+t)$ Subset Sum algorithm \cite{DBLP:conf/soda/Bringmann17}.
 More interestingly, they improved it to $\tilde O(n+\ww^{5/3})$ time, by further incorporating \emph{additive combinatorial techniques} developed by Galil and Margalit \cite{DBLP:journals/siamcomp/GalilM91} and recently refined by Bringmann and Wellnitz \cite{DBLP:conf/soda/BringmannW21}.

 \vspace{0.2cm}
 We remark that in the easier \emph{unbounded} setting 
	(where each item has infinitely many copies available),
 both Question 1 and Question 2 are already resolved, with matching fine-grained lower bounds known.  The 0-1 setting that we consider in our paper seems much more difficult. See the further related works section.

 \subsection{Our contribution}

Now we introduce our new results, which make progress towards resolving Question 1 and Question 2.

\begin{theorem}
	\label{thm:knapsack-main}
The 0-1 Knapsack problem can be solved by a deterministic algorithm with time complexity  \[O(  n + \ww^{2.5}\polylog\ww).\]  Alternatively, it can also be solved by a deterministic algorithm with time complexity
\[O(n\ww^{1.5}\polylog\ww). \]
\end{theorem}

This improves the previous $O(n+\ww^3)$ algorithm \cite{icalp21}. In addition, it also improves the $O(n+\ww t) \le O(n\ww^2)$ time algorithm \cite{DBLP:journals/jco/KellererP04,DBLP:conf/stoc/BateniHSS18,DBLP:conf/icalp/AxiotisT19} in the regime of $t\gg n\sqrt{\ww}$.

By a reduction described in \cite[Section 4]{icalp21}, we have the following corollary which parameterize the running time by the largest item profit $\pp$ instead of $\ww$.
\begin{corollary}
	\label{cor:knapsack-pmax}
The 0-1 Knapsack problem can be solved by a deterministic algorithm with time complexity  \[O(  n + \pp^{2.5}\polylog\pp).\]  Alternatively, it can also be solved by a deterministic algorithm with time complexity
\[O(n\pp^{1.5}\polylog\pp). \]
\end{corollary}

As our secondary result, we improve (and simplify) the subset sum algorithm of \cite{icalp21}.

\begin{theorem}
	\label{thm:subsetsum-main}
The Subset Sum problem can be solved by a randomized algorithm with time complexity  \[O(  n + \ww^{1.5}\log^2\ww)\]  
and $O(1/\poly(\ww))$ error probability.
\end{theorem}

\paragraph*{Updated version (Aug 2023).}
This manuscript is superseded by an improved version posted to arXiv\footnote{\url{https://arxiv.org/abs/2308.04093}} on Aug 8, 2023.
The improved version contains a $0$-$1$ Knapsack algorithm in $\tilde O(n+\ww^2)$ time, which improves over the $\tilde O(n+\ww^{2.5})$ algorithm in this manuscript and the $\tilde O(n+\ww^{2.4})$ algorithm in the recent work of Chen, Lian, Mao, and Zhang (arXiv\footnote{\url{https://arxiv.org/abs/2307.12582}}, July 24, 2023).
The improved version incorporates some content from this manuscript, and also builds on the work of Chen et al. 
The $\tilde O(n+\ww^{3/2})$-time subset sum algorithm and the $\tilde O(n\ww^{1.5})$-time $0$-$1$ knapsack algorithm results in this manuscript are also obtained by Chen et al.\footnote{The $\tilde O(n\ww^{1.5})$ running time for $0$-$1$ Knapsack is implicit in their paper.}, and are omitted in our improved version. This manuscript is now obsolete.

\subsection{Technical overview}
Our Knapsack and Subset Sum algorithms combine techniques from additive combinatorics and integer programming, as well as some algorithmic ideas.
Before describing our new ideas, we first review some known techniques that have been applied in the literature.

\paragraph*{Long progressions in dense subset sums}
There is a long line of work on designing Subset Sum algorithms using techniques from additive combinatorics,  stemming from a series of results by combinatorialists and number theorists in the 80's (e.g., \cite{DBLP:journals/siamcomp/GalilM91, cfgalgo,chaimovichalgo,freimanalgo,freimanparti});
see the survey of Chaimovich \cite{chaimovichsurvey}. These  techniques have also been applied in very recent works on Subset Sum and Knapsack \cite{DBLP:journals/talg/KoiliarisX19,DBLP:conf/soda/MuchaW019,
DBLP:conf/soda/BringmannW21,
icalp21,soda2023knapsack}.

 Ultimately, these algorithms directly or indirectly rely on the following powerful result in additive combinatorics, pioneered by Freiman \cite{freiman93} and \sarkozy \cite{Sarkozy2} and tightened by Szemer\'{e}di and Vu \cite{szemeredivu}, and more recently strengthened by Conlon, Fox, and Pham \cite{fox}:
Let $\caS(A) = \{\sum_{b\in B} b: B\subseteq A\}$ denote the subset sums of $A$. Then,
 \begin{center}
	\textit{If set $A \subseteq  [N]$ has size $|A| \gg \sqrt{N} $, then $\caS(A)$ contains an arithmetic progression  of length $N$.}
 \end{center}

Galil and Margalit \cite{DBLP:journals/siamcomp/GalilM91} implicitly proved a version of this theorem, and used it to derive
(for any $\Omega(\sqrt{N})$-dense input set)
a complete characterization of attainable subset sums in a very long interval, so that the subset sum problem with target number $t$ in that interval can be quickly solved without performing dynamic programming.
This approach was called ``structural characterization'' in the survey of Chaimovich \cite{chaimovichsurvey}.
The algorithm of Galil and Margalit was later refined by Bringmann and Wellnitz \cite{DBLP:conf/soda/BringmannW21}.

We remark that there is also a shortcoming of  this theorem: it completely breaks if $A$ is allowed to be a multiset.  For this reason, some of the earlier algorithms that used this theorem (such as  \cite{DBLP:journals/siamcomp/GalilM91}) had to impose the unnatural restriction that the Subset Sum input set contains no duplicates.
More recent works such as \cite{DBLP:conf/soda/BringmannW21} and \cite{icalp21} relaxed this requirement and allows multiset input, at the cost of slightly increased complexity.

\paragraph*{Proximity technique}
Now we look at another technique from the integer programming literature, the \emph{proximity technique}. In an integer linear program, a proximity result refers to a distance upper bound between its optimal integral solution and the optimal solution of its linear relaxation (i.e., an optimal fractional solution). See e.g., \cite{DBLP:journals/mp/CookGST86,DBLP:journals/talg/EisenbrandW20}. A recent result of Eisenbrand and Weismantel \cite{DBLP:journals/talg/EisenbrandW20} obtained improved proximity results for integer programs using the Steinitz lemma, and also demonstrated new algorithmic applications. 
Later, Polak, Rohwedder, and W\k{e}grzycki
\cite{icalp21} specialized their proximity results to the Knapsack case (1-dimensional integer linear program). In the Knapsack setting, the fractional optimal solution corresponds to the greedy solution (sort in decreasing order of efficiencies $p_i/w_i$, and take the maximal prefix without violating the capacity constraint), and \cite{icalp21} exploited the fact that an optimal knapsack solution differs from the greedy solution by at most $O(\ww)$, in the $\ell_1$ norm.
This proximity bound allows them to shrink the DP array length from $t$ down to $O(\ww^2)$. Then, an $O(n+\ww^3)$ running time was achieved by updating items of the same weight $w$ in a batch, using the SMAWK algorithm (\cite{smawk}).

\vspace{0.5cm}
Now we are ready to introduce our new conceptual ideas. 
\paragraph*{An ``$\ell_0$-proximity''} Our algorithm also (indirectly) relies on the aforementioned theorem on long arithmetic progressions in dense subset sums. 
However, instead of the ``structural characterization'' approach employed by  previous subset sum algorithms (e.g., \cite{DBLP:journals/siamcomp/GalilM91,DBLP:conf/soda/BringmannW21}), we use another additive combinatorial theorem due to \erdos and \sarkozy \cite{erdos-sarkozy} (which was in turn proved using the long-AP theorems). This \erdos-\sarkozy theorem states that: 
for two subsets $A,B \subseteq [N]$ of the same size $|A|=|B|=k$ such that $A$ and $B$ have no common non-zero subset sums, the largest possible size $k$ is at most $k\le O(\sqrt{N\log N})$.
 (This bound was recently tightened to $k\le O(\sqrt{N})$ by Conlon, Fox, and Pham \cite{fox}.)

 Our observation is that, by a minor adaptation of the proof of \erdos and \sarkozy, we can show the following ``$\ell_0$-proximity'' bound for the knapsack problem: if we look at the symmetric difference between a greedy knapsack solution and an optimal knapsack solution, then it only contains items with at most $\tilde O(\sqrt{\ww})$ types of weights. (If we reformulate the knapsack problem as finding a solution vector $\vec x \in \Z_{\ge 0}^{\ww}$ in which $x_w$ denotes the number of weight-$w$ items taken, then we are saying that the $\ell_0$ norm between the greedy solution vector and the optimal solution vector is at most $\tilde O(\sqrt{\ww})$)

 Although the ingredients behind this $\ell_0$-proximity are well known results coming from the fields of additive combinatorics and integer programming, the combination of them seems new, and turns out to be very powerful.  In \cref{sec:subsetsum}, we show how a simple binary-bundling idea combined with this $\ell_0$-proximity bound can lead to a Subset Sum algorithm in $\tilde O(n+\ww^{1.5})$ time, improving (and simplifying) the previous $\tilde O(n+\ww^{5/3})$ algorithm by Polak, Rohwedder, and W\k{e}grzycki
\cite{icalp21}.

\vspace{0.5cm}

Now, we discuss our main result, an $\tilde O( n + \ww^{2.5})$ time algorithm for Knapsack (\cref{thm:knapsack-main}).

\paragraph*{Transfer techniques from unbounded setting}
We take a slight detour to look at some previous results on \emph{unbounded} knapsack/subset sum problems, where each item has  infinitely many copies available. 
This also corresponds to the problem of optimizing over $\Z_{\ge 0}^n$ in the integer linear programming literature.
Usually this setting is much easier, for two main reasons:
\begin{enumerate}
	\item  Since there are infinite supply of items, one does not need to keep track of which items are used so far in the DP.
	\item There are more structural results available, in particular the \emph{Carath{\'{e}}odory-type theorems} \cite{DBLP:journals/orl/EisenbrandS06,DBLP:conf/soda/Klein22}, which show the existence of optimal solution vectors with very small support size (i.e., $\ell_0$ norm), usually  only logarithmic.
\end{enumerate}
This Carath{\'{e}}odory-type bound was very recently used by Deng, Mao, Zhong \cite{dmz23} to design near-optimal algorithms for several unbounded-knapsack-type problems. The main technique introduced there is termed ``witness propagation''.  The idea is that, since the optimal solutions must have small support size (but possibly with high multiplicity), one can first prepare the ``base solutions'', which are partial solutions with small support and multiplicity at most one. Then, they gradually build full solutions from these base solutions, by ``propagating the witnesses'' (that is, increase the multiplicity of some non-zero coordinate). The time complexity of this plan is low since the support sizes are small.

In our 0-1 setting, with a $\ell_0$-proximity bound, it is tempting to also analogously apply the witness propagation framework from \cite{dmz23}. Indeed, this is what we do in our algorithm. However, we still need to overcome several difficulties that arise from the huge difference between 0-1 setting and unbounded setting (in particular, in 0-1 setting we no longer have the convenient property 1 mentioned above).

\paragraph*{Algorithmic ingredients}
We do not give a very detailed description of the algorithm here. Instead, we briefly and informally mention one of the most interesting building block of our algorithm, namely how to perform ``witness propagation'' when every base solution currently in the DP table only has support size equal to one.

We use the SMAWK algorithm to extend from these base solutions (if the support is $\{w\}$, we extend by adding integer multiples of $w$ to the total weight). 
However, since these supports may contain different types of weights $w$, we need to separately deal with them. This means that each type of weight $w$ may only have sublinear number of presences in the entire DP table. Hence, we need to do SMAWK for each of them in sublinear time (and return implicit output, partitioned into segments). This is an very interesting scenario where we actually need to use the tall-matrix version of SMAWK.

Then, we need to update the segments returned by these SMAWK algorithms back to the DP table. This creates another issue, since although the number of these segments is small, the total length of them could still be very large, and we do not know how to update them all onto the DP table without going through each segment from left to right. To solve this issue, we design a novel skipping technique, so that we can ignore the suffixes of some of the segments, while still ensuring that we do not lose the optimal solution. The formal justification of this strategy is subtle, and crucially relies on the concavity of the prefix sums of item profits.

Having solved the case with support size $1$, we can extend them to larger support size $\tilde O(\sqrt{\ww})$, using the two-level color-coding technique originally used by Bringmann \cite{DBLP:conf/soda/Bringmann17} in his subset sum algorithm. Here, we have an additional advantage of knowning these sets in advance, so we can actually derandomize Bringamnn's color coding using standard tools.

\subsection{Further related works}
Our knapsack algorithms (and also \cite{icalp21}) parameterize the running time by $n$ together with either the maximum item weight $\ww$ or the maximum item profit $\pp$.
There are some other results in the literature that parameterize by both $\ww$ and $\pp$ simultaneously. One of the earliest such examples is Pisinger's algorithm in $O(n\ww \pp)$ time \cite{DBLP:journals/jal/Pisinger99}.
More recent examples include \cite{DBLP:conf/stoc/BateniHSS18} and \cite{DBLP:conf/icalp/BringmannC22}.

In contrast to our 0-1 setting, the \emph{unbounded} setting (where each item has infinitely many copies available) has also been widely studied in the literature of Knapsack and Subset Sum algorithms, e.g.,  \cite{DBLP:conf/innovations/Lincoln0W20,DBLP:conf/innovations/JansenR19,doi:10.1287/moor.2022.1308,DBLP:conf/icalp/AxiotisT19,DBLP:journals/jcss/ChanH22,DBLP:conf/soda/Klein22,dmz23}.

Recently there has also been a lot of work on approximation algorithms for Knapsack and Subset Sum (and Partition) \cite{DBLP:conf/soda/Chan18a,DBLP:conf/icalp/Jin19,DBLP:conf/soda/MuchaW019,DBLP:conf/icalp/BringmannN21,soda2023knapsack}. 
For example, the fastest known $(1+\eps)$ approximation algorithm for 0-1 Knapsack runs in $\tilde O(n+1/\eps^{2.2})$ time \cite{soda2023knapsack}.
It appears that approximation may be easier than our exact small-weight setting, but we do not know any formal relation.\footnote{If all profits are real values in $[1,2]$, then by rounding to integer multiples of $\eps$, it can be reduced to an exact instance with  small max item profit $\pp = O(1/\varepsilon)$.} Notably, \cite{soda2023knapsack} also used the additive combinatorial results of \cite{DBLP:conf/soda/BringmannW21} to design knapsack approximation algorithms; this was the first application of additive combinatorial technique to knapsack algorithms. However, some of their steps crucially exploit  approximation, and do not carry over to the exact setting.

\subsection{Paper organization}
\cref{sec:prelim} contains useful definitions and algorithmic tools.
Then in \cref{sec:proximity} we review the earlier $\ell_1$ proximity bound for Knapsack, and present the $\ell_0$ proximity derived from  \erdos and \sarkozy's result. 
In \cref{sec:subsetsum}, we use this $\ell_0$ proximity bound to give a simple Subset Sum algorithm.
Then, \cref{sec:knapsack} contains our main technical result, the 0-1 Knapsack algorithm.

\section{Preliminaries}

\label{sec:prelim}

\subsection{Notations and definitions}
We use $\tilde O(f)$ to denote $O(f\polylog f)$.
Let $[n] = \{1,2,\dots,n\}$.  Let $\pm [ n] = \{1,2,\dots,n\} \cup \{-1,-2,\dots,-n\}$.

We will sometimes deal with integer multisets in this paper.
For an integer multiset $X$, and an integer $x$, we use $\mu_X(x)$ to denote the multiplicity of $x$ in $X$. 
For a multiset $X$, the \emph{support} of $X$ is the set of elements it contains, denoted as $\supp(X):=\{x: \mu_X(x)\ge 1\}$.
We say a multiset $X$ is \emph{supported on} $[n]$ (or on $\Z$, etc.) if $\supp(X) \subseteq [n]$ (or $\supp(X) \subseteq \Z$, etc.).
The \emph{size} of $X$ is $|X| = \sum_{x\in \Z}\mu_X(x) $, and the \emph{sum of elements} in $X$ is $\Sigma(X)= \sum_{x\in \Z}x\cdot \mu_X(x)$.
We also denote the maximum multiplicity of $X$ by $\mu_X:= \max_{x} \mu_X(x)$.
For multisets $A,B$ we say $A$ is a subset of $B$ if for all $a\in A$, $\mu_B(a) \ge \mu_A(a)$, and write $A\subseteq B$.
   We write $A \uplus B$ as the union of $A$ and $B$ by adding multiplicities.

For a multiset $X$, the set of all subset sums of $X$ is $\caS(X):= \{ \Sigma(Y): Y\subseteq X\} $. We also define $\caS^*(X):= \{ \Sigma(Y): Y\subseteq X, Y\neq \emptyset\} $ to be the set of subset sums formed by \emph{non-empty} subsets of $X$. 
The \emph{sumset} of two sets $A,B$ is $A+B = \{a+b : a\in A, b\in B\}$. Their \emph{difference set} is $A-B = \{a-b:a\in A, b\in B\}$.
Let $-A = \{-a: a\in A\}$.

We will work with vectors in $\Z^{\caI}$ where $\caI$ is some index set. We sometimes denote vectors in boldface, e.g., $\vec x\in \Z^\caI$, and use non-boldface with subscript to denote its coordinate, e.g., $x_i \in \Z$ (for $i\in \caI$).
Let $\supp(\vec x):= \{i\in \caI: x_i \neq 0\}$,
$\|\vec x\|_0 :=  \lvert  \supp(\vec x)  \rvert$,
and 
$\|\vec x\|_1 := \sum_{i\in \caI}|x_i|$.
Let $\vec{0}$ denote the zero vector.
Let $\vec{1}$ denote the all-$1$ vector.
For $i\in \caI$, let $\vec e_i$ denote the unit vector with $i$-th coordinate being $1$ and the remaining coordinates being $0$.
We write $\vec x \le \vec y$ iff $x_i\le y_i$ for all $i\in I$.

We use $A[\ell\dd r]$ to denote an array indexed by integers $i\in \{\ell,\ell+1,\dots,r\}$. The $i$-th entry of the array is $A[i]$. Sometimes we consider arrays of vectors, denoted by ${\vec x}[\ell \dd r]$, in which every entry ${\vec x}[i] \in \Z^{\caI}$ is a vector, and we use $x[i]_j$ to denote the $j$-th coordinate of the vector ${\vec x}[i]$ (for $j\in \caI$).

We use the standard word-RAM computation model with $\Theta(\log n)$-bit words.
We assume all input integers, such as 
$t,w_i,p_i$,
 have binary representations that each fit into a single machine word. If this assumption is dropped, the time complexity of our algorithms can still be analyzed by taking into account the arithmetic operations on big integers.

\subsection{Computing sumsets and subset sums}
Given two integer sets $A,B\subseteq \{0,1,\dots,t\}$, we can compute their sumset $A+B$ in $O(t\log t)$ time using Fast Fourier Transform (FFT).

We need  the near-linear time Subset Sum algorithm first developed by Bringmann \cite{DBLP:conf/soda/Bringmann17} and later improved by Jin and Wu \cite{DBLP:conf/soda/JinW19} in terms of logarithmic factors. A useful feature of this algorithm is that it actually reports \emph{all} the attainable subset sums bounded by $t$.
\begin{lemma}[Near-linear algorithm for Subset Sum, \cite{DBLP:conf/soda/Bringmann17,DBLP:conf/soda/JinW19}]
	\label{lem:linear-subset-sum}
	Given a multiset $A$ of $n$ non-negative integers, and an upper bound $t \in \Z^+$, one can compute $\caS(A) \cap [0,t]$ by a randomized algorithm with time complexity $O(n + t \log t)$ and error probability at most $1/t$.
\end{lemma}

\subsection{SMAWK algorithm}
We review the classic result of  Aggarwal,  Klawe, Moran,  Shor, and  Wilber \cite{smawk} on finding row maxima in convex totally monotone matrices (in particular, convex Monge matrices). 

We say an $m\times n$ real matrix $A$ is \emph{convex Monge} if 
\begin{equation}
    \label{eqn:monge}
    A[i,j]+A[i',j'] \ge A[i,j']+A[i',j]
\end{equation} for all $i<i'$ and $j<j'$. 

More generally, we consider matrices with possibly $-\infty$ entries. 
Following the terminology of \cite{DBLP:journals/dam/AggarwalK90}, a matrix $A \in (\R \cup \{-\infty\})^{m\times n}$ is called a \emph{reverse falling staircase} matrix, if the finite entries in each row form a prefix, and the finite entries in each column form a suffix. In other words, if $A[i,j']>-\infty$, then $A[i',j]>-\infty$ for all $i'\in \{i,i+1,\dots,m\}$ and $j \in \{1,2,\dots, j'\}$. In the definition of convex Monge property, inequality \eqref{eqn:monge} is always considered to hold when its right-hand side evaluates to $-\infty$.

A reverse falling staircase matrix $A$ is convex Monge implies $A$ is \emph{convex totally monotone}: for all $i<i'$ and $j<j'$,
\begin{equation}
    \label{eqn:totalmono}
 A[i,j]< A[i,j'] \Rightarrow A[i',j]< A[i',j'].
\end{equation}

Given an $m\times n$ convex totally monotone matrix, let $j_{\max}(i)$ denote the index of the leftmost column containing the maximum value in row $i$.
Note that condition~\eqref{eqn:totalmono} implies 
\begin{equation}
    \label{eqn:rowmaximamono}
1\le     j_{\max}(1) \le j_{\max}(2) \le \dots \le j_{\max}(m) \le n.
\end{equation}
The SMAWK algorithm \cite{smawk} finds $j_{\max}(i)$ for all $1\le i\le m$.  On tall matrices ($n\ll m$),
 its time complexity is near-linear in $n$ (instead of $m$), if we allow a compact output representation  based on \eqref{eqn:rowmaximamono}.
This is formally  summarized in the following theorem.
\begin{theorem}[SMAWK algorithm \cite{smawk}]
    \label{thm:smawk}
   Let an  $m\times n$ convex Monge reverse falling staircase matrix $A$ be implicitly given, so that each entry of $A$ can be accessed in constant time. 
   
   There is a deterministic algorithm that finds all row maxima of $A$ in $O\left (n   (1+\log \left \lceil \frac{m}{n}\right \rceil )\right  )$ time. Its output is compactly represented as $n+1$ integers, $1=  _1\le r_2\le \dots \le r_n \le r_{n+1} = m+1$, indicating that for all $1\le j\le n$ and $r_j\le i <r_{j+1}$,   the leftmost maximum element in row $i$ of $A$ is $A[i,j]$.
\end{theorem}

\subsection{Derandomization}
We need the celebrated deterministic algorithm for finding a coloring of a set system that achieves small discrepancy, using the method of conditional probabilities or pessimistic estimators.

\begin{theorem}[Deterministic set balancing \cite{spencer1987ten,DBLP:journals/jcss/Raghavan88}]
    \label{thm:raghavan}
   Given sets $S_1,S_2,\dots, S_m\subseteq [n]$, there is an $O(nm)$-time deterministic algorithm that finds $x\in \{+1,-1\}^n$, such that for every $i\in [m]$, 
   \[ \left \lvert \sum_{j\in S_i}x_j\right \rvert \le 2\sqrt{|S_i|\ln(2m)}.\]
\end{theorem}
The statement above is taken from the textbook of Chazelle \cite[Section 1.1]{chazelle}. This algorithm works in the word-RAM model with $\Theta(\log (n+m))$-bit words, by performing arithmetic operations with relative error $1/\poly(nm)$ when computing the pessimistic estimators.

Our application requires the following multi-color version of \cref{thm:raghavan}. It is obtained by recursively applying \cref{thm:raghavan}.
\begin{lemma}[Deterministic balls-and-bins]
    \label{thm:detballs}
    Given integer $r$, and $m$ sets $S_1,S_2,\dots, S_m\subseteq [n]$ such that $|S_i|\le r\log m$ for all $i$, there is an $O(nm\log r)$-time deterministic algorithm that finds an $r$-coloring $C\colon [n]\to [r]$, such that for every $i\in [m]$ and every color $c\in [r]$,
    \[ |\{j \in S_i: C(j)=c\}| \le O(\log m).\]
\end{lemma}
\begin{proof}
    Without loss of generality we can assume $r$ is a power of two by decreasing $r$ and increasing the size upper bound to $|S_i| \le b_0:= 2r\log m$. 
    
    We use the following recursive algorithm with $\log_2(r)$ levels: given sets $S_1,\dots,S_m$ in the universe $[n]$ with size bound $b_0$, use \cref{thm:raghavan} to find  a two-coloring $C_0\colon [n] \to \{+1,-1\}$, and then recurse on two subproblems whose $m$ sets are restricted to the universe $C_0^{-1}(+1)$ and the universe $C_0^{-1}(-1)$ respectively. The discrepancy of $C_0$ provided by \cref{thm:raghavan} ensures that the two subproblems contain sets each of size at most $b_0/2 + \sqrt{b_0\ln(2m)}$. We keep recursing in this manner, and the $k$-th level of the recursion tree gives a $2^k$-coloring of the universe $[n]$. The final level gives the desired $r$-coloring of $[n]$.
     By induction, the maximum size of $S_i$ intersecting any color class is at most $b_{\log_2(r)}$, defined by
    \[ b_0 = 2r\log m,\hspace{0.5cm} b_k = b_{k-1}/2 + \sqrt{b_{k-1} \ln(2m)}\hspace{0.2cm} (1\le k\le \log_2(r)).\]
    Then,
\begin{align*}
     \frac{b_k}{b_{k-1}} &\le \frac{1}{2}\exp\left (2 \sqrt{\ln (2m)/b_{k-1}} \right )\\
     & \le \frac{1}{2}\exp\Big (2 \sqrt{2^{k-1}\ln (2m)/b_{0}} \Big ) \tag{using $b_j\ge b_{j-1}/2$},
\end{align*}
and by telescoping,
\begin{align*}
    b_{\log_2(r)} &\le b_0 \prod_{k=1}^{\log_2(r)}\frac{1}{2}\exp\Big (2 \sqrt{2^{k-1}\ln (2m)/b_{0}} \Big ) \\
& = \frac{b_0}{2^{\log_2(r)}}\exp\Big ( 2 \sqrt{\ln (2m)/b_{0}}\sum_{k=1}^{\log_2(r)}\sqrt{2^{k-1}}\Big )\\
& = \tfrac{2r\log m}{r}\exp\Big ( 2\sqrt{\tfrac{\ln(2m)}{2r\log m}} \tfrac{\sqrt{r} - 1}{\sqrt{2}-1}\Big)\\
& = O(\log m).
\end{align*}

    The total time complexity at each level of the recursion tree is $O(nm)$. Summing over all $\log_2(r)$ levels gives $O(nm\log r)$ time complexity overall.
\end{proof}

We also need a pairwise independent hash family with optimal seed length. A hash family $\caH \subseteq \{h\colon [n] \to [m]\}$ is called pairwise independent if for any $x_1\neq x_2\in [n]$ and any $y_1,y_2\in [m]$, $\Pr_{h\in \caH}[h(x_1)=y_1\text{ and } h(x_2)=y_2] = 1/m^2$.
\begin{lemma}[{\cite[Theorem 3.26]{DBLP:journals/fttcs/Vadhan12}}]
    \label{lem:pairwise}
   Let $n\ge m\ge 2$ be powers of two.  There is an explicit family of pairwise independent functions  $\{h \colon [n] \to  [m]\}$ with seed length $\log_2(nm)$ bits. Sampling and evaluating  $h$ takes $O(1)$ time (in word RAM with $\Theta(\log n)$ bits).
\end{lemma}

\section{Proximity technique meets additive combinatorics}
\label{sec:proximity}
In this section, we review the proximity arguments used in previous algorithms for Knapsack and Subset Sum, and then derive a different type of proximity result (\cref{lem:l0l1-prox}) from a slight variant of \erdos and \sarkozy's theorem \cite{erdos-sarkozy}.
The ingredients in the proof of this proximity result are not original; they come from two existing proofs from the literatures of Integer Linear Programming and Additive Combinatorics. However, the combination of them is new, and plays a key role in our new algorithms.

\subsection{Maximal prefix solution for Knapsack}
\label{subsec:maximalprefix}
We can assume $\ww \le t$ by ignoring items that are too large to fit into the  knapsack. And, we assume 
\begin{equation}
    \label{eqn:nontrivial-assumption}
    w_1+\dots + w_n>t,
\end{equation} since otherwise the trivial optimal solution is to include all the items.

Sort the $n$ items $\{(w_i,p_i)\}_{i=1}^n$ in decreasing order of \emph{efficiency} (profit-to-weight ratio), \[p_1/w_1 \ge p_2/w_2 \ge \dots \ge p_n/w_n,\] breaking ties arbitrarily.
Then, the \emph{maximal prefix solution} is the item subset $P=\{1,2,\dots, i^*\}$ where 
\[i^* = \max\{ i^* : w_1+w_2+\dots+w_{i^*} \le t\},\] i.e.,  we greedily take the most efficient items one by one, until the next item cannot be added without exceeding the knapsack capacity.

In the following, for any subset $I = \{i_1,\dots,i_{|S|}\}\subseteq [n]$ of items, let $\wts(I) := \{w_{i_1},\dots,w_{i_{|S|}}\}$ denote the multiset of the item weights in $I$, let $w_I:= \sum_{i\in I} w_i$ denote the total weight of $I$, and similarly let $p_I$ denote the total profit of $I$.

By assumption \cref{eqn:nontrivial-assumption},  the maximal prefix solution $P$ must have total weight satisfying 
\begin{equation}
    \label{eqn:wtotal-range}
    w_P \in (t - \ww,t],
\end{equation}
since otherwise one can always add the next item without exceeding the capacity $t$.
Similarly, any optimal knapsack solution $Q\subseteq [n]$ should also satisfy
\begin{equation}
    \label{eqn:wtotal-range-optimal}
    w_Q \in (t - \ww,t].
\end{equation}

The following property of maximal prefix solution follows from a simple exchange argument.
\begin{lemma}[No common subset sum]
    \label{lem:no-common-subsetsum}
    Let $P = \{1,2,\dots,i^*\} \subseteq [n]$ denote the maximal prefix solution. Then there exist an optimal knapsack solution $Q\subseteq [n]$ 
   such that  
    \begin{equation}
        \label{eqn:no-common-subsetsum}
    \caS^*(\wts(P\setminus Q))\cap \caS^*(\wts(Q\setminus P)) =\emptyset.
    \end{equation}
    Moreover, if $p_{i^*}/w_{i^*}> p_{i^*+1}/w_{i^*+1}$, then \emph{all} optimal knapsack solutions satisfy \cref{eqn:no-common-subsetsum}.
\end{lemma}
\begin{proof}
   Let $Q\subseteq [n]$ be any optimal knapsack solution.  Denote $A := P \setminus Q, B := Q\setminus P$, and hence $Q =(P \setminus A) \cup B$. 
   
   Suppose for contradiction that \cref{eqn:no-common-subsetsum} fails. 
   Then $\wts(A)$ and $\wts(B)$ have some common subset sum $s =w_{A'}=w_{B'}$ realized by non-empty subsets $A'\subseteq A, B'\subseteq B$. 
   Now define another solution \[Q':= (Q\cup A')\setminus B',\] which is feasible because 
   \[ w_{Q'} = w_{Q} + w_{A'} - w_{B'} = w_{Q}\le t,\]
   and achieves a total profit gain 
   \begin{align*}
     p_{Q'}-p_{Q} &=  p_{A'} - p_{B'} \\
     & \ge w_{A'}\cdot \min_{i\in A'}\frac{p_i}{w_i} - w_{B'}\cdot \max_{i\in B'}\frac{p_i}{w_i} \\
     & = s\cdot \left (\min_{i\in A'}\frac{p_i}{w_i} -  \max_{i\in B'}\frac{p_i}{w_i}\right )\\
     & \ge 0,
   \end{align*}
   where the last step follows from 
   $A'\subseteq A\subseteq P=\{1,2,\dots,i^*\}$ and $B'\subseteq B\subseteq [n]\setminus P = \{i^*+1,\dots,n\}$, and 
   the decreasing order of $p_i/w_i$. Hence, $Q'$ is no worse than the optimal solution $Q$, and should also be optimal. 

   In the case where $p_{i^*}/w_{i^*}> p_{i^*+1}/w_{i^*+1}$, the above inequality is strict, and we already have a contradiction to the optimality of $Q$. Hence, $Q$ must satisfy \cref{eqn:no-common-subsetsum}.

In the case without this condition, we can additionally require $Q$ to also minimize $|P\setminus Q|$.  Then, 
   \[ |P\setminus Q'| = |P \setminus (Q\cup A')| = |P\setminus Q| - |A'| < |P\setminus Q|,\]
   contradicting the minimality of $|P\setminus Q|$. Hence, $Q$ must satisfy \cref{eqn:no-common-subsetsum}.
\end{proof}
Previous works \cite{DBLP:journals/talg/EisenbrandW20,icalp21} applied \cref{lem:no-common-subsetsum} in combination with the following simple fact (with $\ww$ in place of $N$), for which we include a proof for completeness.
\begin{lemma}
    \label{lem:size-bound}
  If two multisets $A,B$ supported on $[N]$ satisfy   
$\caS^*(A) \cap \caS^*(B) = \emptyset$ and $|\Sigma(A)- \Sigma(B)|<N$, then
\[ |A|+|B|\le 2N.\]
\end{lemma}
\begin{proof}
    Let $C = A\uplus (-B)$, which is a multiset supported on $[-N,N] \cap \Z$ with element sum $\Sigma(C) = \Sigma(A)-\Sigma(B) \in [-N,N]$. 
    We find a permutation $(c_1,c_2,\dots,c_{|C|})$ of the elements of $C$, so that 
    \begin{equation}
        \label{eqn:suffix-range}
        c_{i+1}+c_{i+2}+\dots+c_{|C|} \in [-N,N]
    \end{equation}
 holds for all $0\le i \le |C|$.
  This can be inductively constructed in the following greedy fashion: after fixing the first $i$ integers $(c_1,\dots,c_i)$, let $C'$ denote the multiset of the remaining $|C|-i$ integers, and choose an arbitrary $c_{i+1}\in C'$ that has sign agreeing with the sign of $\Sigma(C')$, which can maintain the invariant \[|c_{i+2}+\dots + c_{|C|}| = |\Sigma(C') - c_{i+1}| = \big \lvert |\Sigma(C')| - |c_{i+1}| \big \rvert \le \max\{ |\Sigma(C')| , |c_{i+1}|\} \le  N.\]

Suppose for contradiction that $|C|+1  >2N+1$.  Then by \cref{eqn:suffix-range} and the pigeonhole principle, there exist $0\le i<j \le |C|$ such that $c_{i+1}+\dots + c_{|C|} = c_{j+1}+\dots + c_{|C|}$, or equivalently, 
\begin{equation}
    \label{eqn:csum0}
 c_{i+1}+\dots +c_j = 0.
\end{equation}
 According to the definition of $C$, \cref{eqn:csum0} corresponds back to a pair of equal subset sums between $A$ and $B$, contradicting the assumption $\caS^*(A) \cap \caS^*(B) = \emptyset$.
Hence, we must have $2N\ge |C| = |A|+|B|$.
\end{proof}
The discussion above is simply a rephrasing of the arguments in \cite{icalp21}, which were in turn based on the work of Eisenbrand and Weismantel \cite{DBLP:journals/talg/EisenbrandW20} which showed proximity for general integer linear programs using the Steinitz lemma.  
A similar idea also appeared in the earlier work of Pisinger \cite{DBLP:journals/jal/Pisinger99}.

\cite{icalp21} noted that the maximal prefix solution can be found in $O(n)$ time using linear-time median finding algorithms, as opposed to a straightforward $O(n\log n)$-time sorting. 

\subsection{A result of \erdos and \sarkozy}
\label{subsec:erdossarkozy}
When studying non-averaging sets \cite{MR0316255,MR0316256}, \erdos  and \sarkozy \cite{erdos-sarkozy} showed that for two subsets $A,B \subseteq [N]$ with $|A|=|B|=k$ such that $\caS^*(A)\cap \caS^*(B) =\emptyset$, the largest size $k$ is at most $k\le O(\sqrt{N\log N})$.
The proof relies on a result of \sarkozy  \cite{Sarkozy2} (independently proved by Freiman \cite{freiman93}).
 This bound was recently tightened to $k\le O(\sqrt{N})$ by Conlon, Fox, and Pham \cite{fox}.
An arithmetic progression $a,a+d,a+2d,\dots,a+(k-1)d$ ($d\neq 0$) is called a \emph{homogenous progression} if $d$ divides $a$ (and hence all elements in the progression).

 \begin{theorem}[\cite{fox}, improving \cite{freiman93} and \cite{Sarkozy2}]
	\label{thm:arith}
	 There is a constant $c$ such that if set $A \subseteq  [N]$ with $|A| \ge c\sqrt{N} $, then $\caS(A)$ contains a homogeneous progression of length $N$.
 \end{theorem}

 This theorem is optimal up to the constant $c$, as witnessed by the example $A = \{1,2,\dots,\lfloor \sqrt{2N-2} - 1/2\rfloor\}$.

Szemer\'{e}di and Vu \cite{szemeredivu} proved the existence of a possibly non-homogenous arithmetic progression. The result of \cite{fox} is a common strengthening of \cite{freiman93,Sarkozy2} and \cite{szemeredivu}.

Galil and Margalit \cite{DBLP:journals/siamcomp/GalilM91} proved the existence of a long arithmetic progression assuming a stronger density requirement $|A| \ge \Omega(\sqrt{N} \log N)$. However, their proof is explicit and algorithmic, in the sense that for every element in the constructed progression, one can efficiently find a subset achieving that subset sum.

Bringmann and Wellnitz \cite{DBLP:conf/soda/BringmannW21} refined Galil and Margalit \cite{DBLP:journals/siamcomp/GalilM91}'s dense subset sum algorithm by improving their running time, and also extending the algorithm to multisets.

We need an analogous but slightly different statement, with a precondition of $\Sigma(A)\approx \Sigma(B)$ instead of $|A|=|B|$.
The proof is analogous to \cite{erdos-sarkozy} (and even simpler).

\begin{lemma}
	\label{lem:support}
	There is a constant $C$ such that if two multisets $A,B$ supported on  $[N]$ satisfy $\caS^*(A) \cap \caS^*(B) = \emptyset$ and $|\Sigma(A)- \Sigma(B)|<N$, then  
    \begin{equation}
        \label{eqn:supp}
        \max \big \{ \lvert \supp(A)\rvert , \lvert \supp(B)\rvert \big \}\le C \sqrt{N}.
    \end{equation}
\end{lemma}
\begin{proof}
	The proof essentially follows \cite{erdos-sarkozy}. Let $C = 2c$, where $c$ is the constant from \cref{thm:arith}.
	Suppose the claimed inequality \eqref{eqn:supp} does not hold. 
	Without loss of generality, assume $\lvert \supp(A) \rvert \ge \lvert \supp(B) \rvert$, and let $A' = \{a_1,a_2,\dots, a_m\}$ be the set of distinct elements in $A$, where $m = \lvert \supp(A) \rvert > C\sqrt{N} > c\sqrt{N+1}$.  We can apply \cref{thm:arith} to set $A'\subset [N+1]$, and obtain a $(N+1)$-term homogeneous progression  
\begin{equation}
	\label{eqn:homop}
		\{(i-N)d, \dots,(i-1)d, id \} \subseteq \caS(A') \subseteq \caS(A),
\end{equation}	
 where $d\ge 1$ and $i \ge N$.

Pick the maximum $r$ such that  
\begin{equation}
	\label{eqn:rd}
	r\in \caS(B) \cap \{0,d,2d,\dots,id\}.
\end{equation}
	 Such $r$ exists since $0\in \caS(B)$.    If $r>(i-N)d$, then $r>0$ and $r$ is contained in the progression \eqref{eqn:homop}, and hence $r\in \caS^*(A)\cap\caS^*(B)$, contradicting $\caS^*(A)\cap\caS^*(B)=\emptyset$. 
Hence, in the following we assume 
\begin{equation}
	\label{eqn:r}
	r\le (i-N)d.
\end{equation}

Suppose subset sum $r\in \caS(B)$ is achieved by subset $R\subseteq B$.  Let $B\setminus R$ denote the complement of $R$ in multiset $B$, which is naturally defined by subtracting the multiplicities of elements.
Combining \eqref{eqn:r} with $|\Sigma(A)-\Sigma(B)| <N$, we obtain
\begin{align*}
	\Sigma(B\setminus R) &= \Sigma(B)-r\\
	&> (\Sigma(A)-N) - (i-N)d\\
	   & = (\Sigma(A) - id) + N(d-1) \\
	   & \ge N(d-1),
\end{align*}
where the last line follows from $id\in \caS(A)$.  Then, since $\supp(B) \subseteq [N]$, we have $|B\setminus R|\ge \Sigma(B\setminus R)/N > d-1$. 

Arbitrarily pick a $d$-subset $\{b_1,b_2,\dots,b_d\}\subseteq B\setminus R$.
By the pigeonhole principle, there exist
 $0\le j<k\le d$ such that $b_1+\dots +b_j \equiv b_1+\dots + b_k \pmod{d}$, and hence $b_{j+1}+\dots +  b_k$ is a multiple of $d$. 
Define multiset $R' = R \uplus \{b_{j+1},\dots,b_k\} \subseteq B$, whose sum \[\Sigma(R') = r + b_{j+1}+\dots +b_k\] is also a multiple of $d$, due to \eqref{eqn:rd}.
 Moreover,  $\Sigma(R') > r$, and
 \begin{align*}
	\Sigma(R') &\le r+N(k-j) \\
	 & \le r+Nd \\
	 & \le id. \tag{due to \eqref{eqn:r}}
 \end{align*}
 Hence, $\Sigma(R')$ is also contained in  $\caS(B) \cap \{0,d,2d,\dots,id\}$. However, this contradicts the maximality of $r$. Therefore,
 the claimed inequality \eqref{eqn:supp} must hold. 
  \end{proof}
If we used the earlier results \cite{freiman93,Sarkozy2} instead of \cite{fox}, we would obtain a weaker bound of $C\sqrt{N\log N}$, which would still be sufficient for our algorithmic applications if we ignore logarithmic factors in the running time.

Finally, we apply \cref{lem:size-bound} and \cref{lem:support} to \cref{lem:no-common-subsetsum} by setting $N:=\ww, A:=\wts(P\setminus Q), B:= \wts(Q\setminus P)$, we obtain the following key lemma.
\begin{lemma}[$\ell_1$ and $\ell_0$ proximity]
    \label{lem:l0l1-prox}
    Let $P = \{1,2,\dots,i^*\} \subseteq [n]$ denote the maximal prefix solution. Then there exist an optimal knapsack solution $Q\subseteq [n]$ 
   such that  
    \begin{equation}
        \label{eqn:leml1prox}
    |P\setminus Q| + |Q\setminus P| \le 2\ww,
    \end{equation}
    and 
    \begin{equation}
        \label{eqn:leml0prox}
    \lvert \supp(\wts(P\setminus Q)) \rvert + \lvert \supp(\wts(Q\setminus P)) \rvert \le 2C\sqrt{\ww},
    \end{equation}
    where $C$ is a universal constant.

    Moreover, if the items have distinct efficiencies, then \emph{all} optimal solutions satisfy  \cref{eqn:leml1prox} and \cref{eqn:leml0prox}.
\end{lemma}

\section{Subset Sum}
\label{sec:subsetsum}

In this section we present the $\tilde O(n+\ww^{1.5})$-time Subset Sum algorithm claimed in \cref{thm:subsetsum-main}.

Given a Subset Sum instance (which is a special case of Knapsack), after computing the maximal prefix solution (\cref{subsec:maximalprefix}) in linear time and applying \cref{lem:l0l1-prox}, we arrive  at the following problem: 
\begin{prob}
Given a multiset $Z$ consisting of $n$ integers in $[-\ww,\ww]$, and a target integer $0\le t^* < \ww$, decide whether there exists a subset $X\subseteq Z$ such that $\Sigma(X) = t^*$. 

Moreover, in the case where it exists, we are promised that there exists such $X$ satisfying $|X| \le  2\ww$ and $\lvert \supp(X) \rvert \le 2C\sqrt{\ww} $.
\label{prob:subsetsum-after}
\end{prob}
For convenience we formulated \cref{prob:subsetsum-after} as a decision problem; the optimization version easily reduces to the decision version by binary search and adding dummy elements.

Now it remains to give an algorithm that solves \cref{prob:subsetsum-after}.  
We can without loss of generality assume that $\mu(Z) \le 2\ww$ by removing duplicates, since the number of copies of any integer that we can use in our solution $X$ is at most  $|X| \le 2\ww$.
In the ideal scenario where the desired multiset $X$ has maximum multiplicity $\mu(X) = O(1)$, the square-root upper bound on $\lvert \supp(X)\rvert$ would translate into an upper bound on $|X|$, which would benefit the dynamic programming algorithm. 
Of course,  we cannot achieve this in general, because the input multiset $Z$ can contain high-multiplicity items. 

This motivates the binary bundling trick, which is standard in the literature (e.g., \cite{martellototh,DBLP:journals/talg/KoiliarisX19}):
every positive integer $k$ can be written as a sum of $\kappa \le 2\log_2 (k+1)$ many powers of 2, 
\[ k = 2^{\alpha_1} + \dots + 2^{\alpha_\kappa}, \;   \alpha_i \in \Z_{\ge 0},\]
such that $ \caS(\{2^{\alpha_1},\dots, 2^{\alpha_\kappa}\}) =  \{0,1,\dots,k\} $, and no integer occurs more than twice among $\alpha_1,\dots,\alpha_\kappa$.\footnote{These powers can be constructed as follows: let $m = \lfloor \log_2(k+1) \rfloor$, and return $A\uplus B$, where $A:= \{2^0,2^1,\dots,2^{m-1}\}$, and $B$ consists of the powers that form the binary representation of $k - (2^m-1)$.}
This allows us to replace each type of input item $z\in Z$ with multiplicity $\mu_Z(z)=k$ by the bundled items $2^{\alpha_1}\cdot z, \dots, 2^{\alpha_\kappa}\cdot z$. 
These new items form a Subset Sum instance, 
\begin{equation}
\hat Z = \hat Z_0\uplus \hat Z_1\uplus \dots \uplus \hat Z_{\ell}, \quad \quad (\ell = \lfloor \log_2(2\ww)\rfloor) 
\label{eqn:subsetsuminstance}
\end{equation}
where $\hat Z_\alpha$ is a multiset supported on $2^\alpha \cdot \{-\ww, \dots,\ww\}$, which consists of those bundled items with coefficient $2^\alpha$. 
This new instance $\hat Z$ has a solution if and only if the original instance $Z$ has one.

This binary bundling trick reduces the multiplicities of items at the cost of blowing up their sizes by powers of two, so we cannot directly apply a small-item subset sum algorithm to this new instance.
 However, as we will see in the following, if we merge the items (i.e., computing the sumsets) in a careful order, then all the intermediate sums will be integer multiples of the blow-up factor, so effectively we can pretend that the sizes did not blow up at all. We remark that this trick was also used in the Subset Sum algorithm by \cite{icalp21}: there, they did two layers of bundling, while here we are able to use the more standard $O(\log \ww)$-level bundling, thanks to the power of $\ell_0$-proximity.

Observe that a solution $X\subseteq Z$ to the original \cref{prob:subsetsum-after} now corresponds to a solution $(\hat X_0,\hat X_1,\dots,\hat X_\ell)$ where
\[ \Sigma(\hat X_0 \uplus \cdots \uplus \hat X_\ell) = t^*,\]
\[ \hat X_\alpha \subseteq \hat Z_\alpha , \] 
\[ \mu(\hat X_\alpha)\le 2, \] 
and
\[ \lvert \supp(\hat X_\alpha)\rvert  \le 2C\sqrt{\ww}\]
for all $0\le \alpha \le \ell$.
The last two conditions imply 
\begin{equation}
 |\hat X_\alpha| \le 4C\sqrt{\ww}
 \label{eqn:xalphasize}
\end{equation}
for all $0\le \alpha \le \ell$.

This gives us the following estimate on the ``high-multiplicity'' part of the any solution 
$(\hat X_0,\hat X_1,\dots,\hat X_\ell)$:
for every $0\le \beta \le \ell$, we have
\begin{align}
	|\Sigma(\hat X_\beta\uplus \hat X_{\beta+1}\uplus \dots \uplus \hat X_{\ell})|  & = | t^* - \Sigma(\hat X_0) - \Sigma(\hat X_1) - \cdots - \Sigma(\hat X_{\beta-1})| \nonumber \\
	& \le |t^* | + \sum_{\alpha=0}^{\beta - 1}| \Sigma(\hat X_\alpha)| \nonumber \\
	&\le  \ww + \sum_{\alpha=0}^{\beta - 1} 4C\sqrt{\ww}\cdot  \max_{x\in \hat X_\alpha}|x| \tag{due to \cref{eqn:xalphasize}}\nonumber\\
	& \le  \ww + \sum_{\alpha=0}^{\beta - 1} 4C\sqrt{\ww}\cdot  (2^\alpha \ww) \nonumber \\
	& \le 5C\cdot 2^\beta\cdot \ww^{1.5}. \label{eqn:suffix-levels-bound}
\end{align}

\begin{algorithm}
\DontPrintSemicolon
\caption{Subset Sum algorithm, given input $\hat Z_0,\dots \hat Z_\ell$ (as in \cref{eqn:subsetsuminstance})}
\label{alg:subsetsum}
Initialize set $S_{\ell+1}  := \{0\}$\\
\For{$\beta  \gets \ell, \ell -1, \dots, 0 $}{
	Partition $\hat Z_\beta = \hat Z_\beta^+ \uplus (-\hat Z_\beta^-)$, where $\hat Z_\beta^{+},\hat Z_\beta^{-}$ are non-negative integer multisets\\
	Compute set $T_{\beta} := T_\beta^+ - T_\beta ^-$, where
	\vspace{-0.4cm}
	\[\hspace{1.96cm}\begin{matrix}T_{\beta}^+ := 2^\beta \cdot \big (\caS(2^{-\beta}\cdot \hat Z_\beta^+) \cap [0,2C \ww^{1.5}] \big ), &\\ T_{\beta}^- := 
		2^\beta \cdot \big (\caS(2^{-\beta}\cdot \hat Z_\beta^-) \cap [0,2C \ww^{1.5}] \big )\phantom{,} &\text{\hspace{-0.35cm}are computed using \cref{lem:linear-subset-sum}}\end{matrix}
	\hspace{1cm}
	\]
	 	\vspace{-0.5cm}
	 \\
	Compute set $S_{\beta} := (S_{\beta+1}  + T_{\beta}) \cap [\pm 5C\cdot 2^\beta \cdot \ww^{1.5} ] $ \\
}
\Return YES if $t^* \in S_0$, NO otherwise
\end{algorithm} 

Now we describe our algorithm to find a solution $(\hat X_0,\hat X_1,\dots,\hat X_\ell)$. A pseudocode is given in \cref{alg:subsetsum}.
It iterates $\beta$ from $\ell$ down to $0$ and computes $S_\beta \subset \Z$ in each iteration, with the invariants that
\begin{equation}
	\label{eqn:invariant-hassol}
S_\beta \ni \Sigma(\hat X_\beta \uplus  \cdots \uplus \hat X_{\ell}),
\end{equation}
and 
\begin{equation}
	\label{eqn:invariant-validsol}
S_\beta \subseteq \caS( \hat Z_{\beta} \uplus \cdots \uplus \hat Z_{\ell}).
\end{equation}
Assuming these invariants hold, eventually the algorithm correctly decides whether a solution exists, by checking whether $t^* \in S_0$.

 In iteration $\beta$, \cref{alg:subsetsum} first   computes a set  $T_\beta \subseteq \caS(\hat Z_\beta)$ that contains $\Sigma(\hat X_\beta)$ (due to \cref{eqn:xalphasize}), and then  
 convolutes $T_\beta$ with the set $S_{\beta+1} \subseteq \caS( \hat Z_{\beta+1} \uplus \cdots \uplus \hat Z_{\ell})$ accumulated from previous rounds, and hence by induction the resulting sumset  $T_\beta + S_{\beta+1}$ contains the desired partial sum $\Sigma(\hat X_\beta) + \Sigma(\hat X_{\beta+1} \uplus\cdots \uplus \hat X_{\ell})$.
 Then, by the upper bound given in \cref{eqn:suffix-levels-bound}, we can safely truncate $T_\beta+ S_{\beta+1}$ down to $[\pm 5C\cdot 2^\beta \cdot \ww^{1.5} ]$ without losing the desired partial sum, hence establishing invariant \cref{eqn:invariant-hassol}. Invariant \cref{eqn:invariant-validsol} also follows easily.

Now we look at the implementation in more detail, and analyze the time complexity. %
Since all integers involved here are multiples of $2^\beta$, the effective length of the array is only $O(2^\beta \ww^{1.5})/2^\beta = O(\ww^{1.5})$. Hence, the application of \cref{lem:linear-subset-sum} only takes $O(|\hat Z_\beta| + \ww^{1.5}\log \ww)$ time, and the convolution step again takes $O(\ww^{1.5}\log \ww)$ time using FFT.

So the overall running time is $\sum_{\beta=0}^\ell O(|\hat Z_\beta| + \ww^{1.5}\log \ww)  = O(n + \ww^{1.5}\log^2 \ww)$. This finishes the proof of \cref{thm:subsetsum-main}.

\section{0-1 Knapsack}
\label{sec:knapsack}
In this section we present the algorithm for 0-1 Knapsack (\cref{thm:knapsack-main}).

\subsection{Reduction via proximity}
\label{subsec:knapsack-prelim}

In this section, we reduce 0-1 Knapsack with maximum item weight $\ww$ to the following \textsc{KnapsackProximity} problem, using arguments from \cref{sec:proximity} and known or standard tricks (some of which already appeared in \cite{icalp21}).  

We say a function $f\colon \Z_{\ge 0} \to \Z$ is \emph{strictly concave}, if $f(x)-f(x-1)> f(x+1)-f(x)$ for all $x\ge 1$.

\begin{prob}[\textsc{KnapsackProximity}]
    \label{prob:transformedknapsack}
Given $\caW \subseteq \pm [\ww]$, and for every $w\in \caW$, 
a strictly concave function $P_w\colon \Z_{\ge 0}\to \Z$ with $P_w(0)=0$ that can be evaluated in constant time, the task is to find $\vec x \in \Z_{\ge 0}^{\caW}$ that maximizes
\[ \sum_{w\in \caW} P_w(x_w)\]
subject to
\[ \sum_{w\in \caW} w\cdot x_w \le t^*.\]
In addition, we are promised that all optimal solutions $\vec x$ satisfy $\|\vec x\|_0 \le b_0$ and $\|\vec x\|_1 \le b_1$, for some given parameters $b_0,b_1$.
\end{prob}

Notice that we imposed several tie-breaking constraints that will be helpful for algorithm design, including the \emph{strict} concavity condition and the promise that \emph{all} optimal solution vectors have bounded $\ell_0$- and $\ell_1$-norms.

\begin{lemma}[Reduction to \textsc{KnapsackProximity}]
    \label{lem:reducetokproximi}
    In $O(n + \min\{n,\ww\}\ww\log \ww)$ time, one can deterministically reduce the
    0-1 Knapsack problem on $n$ items with maximum weight $\ww$ to \textsc{KnapsackProximity} (\cref{prob:transformedknapsack}) with parameters  $b_0 = O(\sqrt{\ww})$, $b_1 = \min\{n,2\ww\}$, and $|\caW| \le \min\{n,2\ww\}$.
\end{lemma}

Now we state the running time of our algorithm for \textsc{KnapsackProximity}.
\begin{theorem}
    \label{thm:solveknapsackprox}
   \textsc{KnapsackProximity} (\cref{prob:transformedknapsack}) can be solved by a deterministic algorithm in 
   \[ \tilde O(b_0\ww (|\caW| + b_1)) \]
   time.
\end{theorem}
Combining \cref{thm:solveknapsackprox} with \cref{lem:reducetokproximi}, the overall running time is $O(n) + \tilde O( \ww^{1.5}\min \{n,\ww\})$, which immediately proves \cref{thm:knapsack-main}. %
 (A more precise time bound for 0-1 Knapsack is $O(n + \min\{\ww^{2.5},n\ww^{1.5}\}\log^3 \ww)$; 
see the end of the proof of \cref{thm:solveknapsackprox}.)
    
The algorithm claimed in \cref{thm:solveknapsackprox} will be given in later sections.  The rest of this section proves the reduction result in \cref{lem:reducetokproximi}.

We first break ties for the profits $w_i$ and efficiencies $p_i/w_i$.
\begin{lemma}[Break ties]
    \label{lem:breakties}
    Given a 0-1 Knapsack instance $I$, in $O(n)$ time we can reduce it to another 0-1 Knapsack instance $I'$ with $n,\ww$ and $t$ unchanged, and $\pp' \le  \poly(\pp,\ww, n)$, such that the items in $I'$ have distinct profits and distinct efficiencies.
\end{lemma}
\begin{proof}
    Suppose instance $I$ has capacity $t$ and $n$ items $(w_1,p_1),\dots,(w_n,p_n)$. Define instance $I'$ with capacity $t$ and items $(w_1,p_1'),\dots,(w_n,p_n')$ with modified profits
    \[ p_i' := (p_i \cdot M + i)\cdot \ww + 1,\]
    where $M:= 1+n+\sum_{i=1}^n i$.
    Then, for any item set $S\subseteq [n]$, 
   we have
   \[ 0 \le \sum_{i\in S}p_i' - M\ww\sum_{i\in S}p_i = |S| + \sum_{i\in S}i\ww   < M\ww,\]
   and hence
\[\sum_{i\in S}p_i  = \left \lfloor \frac{\sum_{i\in S}p_i' }{M\ww}\right \rfloor,\]
    so any optimal solution for $I'$ must also be an optimal solution for $I$.

    For any $i\neq j$, note that $p_i' \bmod (M\ww) = i\ww + 1\neq j\ww + 1 = p_j' \bmod (M\ww)$, so $p_i' \neq p_j'$. 
   If $p_i'/w_i = p_j'/w_j$, then from $p_i'w_j \equiv w_j \pmod{\ww} $ and
   $p_j'w_i \equiv w_i \pmod{\ww} $
    we have $w_j = w_i$, which then contradicts $p_i' \neq p_j'$. So $p_i'/w_i \neq p_j'/w_j$.
\end{proof}

Given a 0-1 Knapsack instance with $n$ items $(w_1,p_1),\dots,(w_n,p_n)$, we can assume these items have distinct profits and efficiencies, by first running the reduction in \cref{lem:breakties}.
 Then, we run the reduction described in \cref{sec:proximity} to find a maximal prefix solution $P\subseteq [n]$, and $t^* = t - \sum_{i\in P}w_i\ge 0$. Denote $Q = [n]\setminus P$.
  Both steps take $O(n)$ time.
Now, the problem becomes to add items from $Q$ to the solution, and discard items from $[n]\setminus Q$. By the $\ell_1$-proximity bound from \cref{lem:l0l1-prox}, the number of added and discarded items does not exceed $2\ww$. It also clearly does not exceed $n$, the total number of items. So we let $b_1 = \min\{n,2\ww\}$.

Denote 
\[\caW^+ = \bigcup_{i\in Q} \{w_i\},\]
\[\caW^- = \bigcup_{i\in [n]\setminus Q} \{-w_i\},\]
and let 
\[ \caW = \caW^+ \cup \caW^-.\]
Note that $\caW \subseteq \pm [\ww]$ and $|\caW| \le n$.

For every $w\in \caW^+$, let $p_{w,1}>p_{w,2}>\dots>p_{w,k_w}>0$ be the profits of the items in set $Q$ that have weight equal to $w$.
Obviously, if the number of weight-$w$ items to add to the solution is fixed to $x$, then we should always take the top $x$ items, with profits $p_{w,1},\dots,p_{w,x}$.
Hence we define 
\[P_w(x):= p_{w,1}+p_{w,2}+\dots + p_{w,x}.\]
We can truncate the items after the $b_1$-th one because they can never be used by any optimal solution, so we can assume $k_w\le b_1 \le 2\ww$.
For convenience, we extend the sequence to $p_{w,x}:= -M-x$ for all $x>k_w$, where integer $M :=  p_1+p_2+\dots + p_n + 1$ is defined to be larger than the total profit of all items.
Hence $P_w(\cdot)$ is defined for all $x\in \Z_{\ge 0}$.  Note that $P_w(\cdot)$ is strictly concave.

Similarly, for every $-w\in \caW^-$, let $0<p_{-w,1}<p_{-w,2}<\dots<p_{-w,k_{-w}}$ be the profits of the items in set $P$ that have weight equal to $w$. We extend the sequence to $p_{-w,x}:= -M-x$ for all $x>k_{-w}$, and define
\[P_{-w}(x):= p_{-w,1}+p_{-w,2}+\dots + p_{-w,x}\]
for all $x\in \Z_{\ge 0}$. Again, $P_{-w}(x)$ is also strictly concave in $x$.

Let $x_w \ge 0$ for each $w \in \caW $ count the number of weight-$w$ items to be added to (or the number of weight-$(-w)$ items to be removed from) the maximal prefix solution. Then the profit gain is $\sum_{w\in \caW}P_w(x_w)$, and the total weight constraint is $\sum_{w\in \caW}w\cdot x_w \le t^*$.
Note that illegal solutions (i.e., with counts greater than the actual number of available items) always have negative objective, due to our choice of $M$ in the definition of $P_w(\cdot)$.
We write the counts as a solution vector $\vec x \in \Z_{\ge 0}^{\caW}$ indexed by $\caW$.
Since all items have distinct efficiencies, \cref{lem:l0l1-prox} gives the bounds $\|\vec x\|_1 \le b_1 = \min\{n,2\ww\}$ and $\|\vec x\|_0 \le b_0 = 2C\sqrt{\ww} $ for \emph{all} optimal solution vectors~$\vec x$.
Hence, we have transformed 0-1 Knapsack exactly to the formulation of \textsc{KnapsackProximity} (\cref{prob:transformedknapsack}).

By preprocessing prefix sum arrays, $P_w(x)$ can be accessed in constant time. This preprocessing takes at most $O(n + |\caW|\ww\log \ww)$ time: for each $w\in \caW$ the top-$k_w$ items (where $k_w\le \ww$) can be selected in linear time, and then sorting them takes $O(\ww\log \ww)$ time. This concludes the proof of \cref{lem:reducetokproximi}.

\subsection{Preparing base solutions}
\label{subsec:knapsack-base}

Now we begin to describe our algorithm for \textsc{KnapsackProximity} (\cref{prob:transformedknapsack}).

Our algorithm uses dynamic programming (DP).
To exploit the $\|\vec x\|_0\le b_0$ bound, we use the \emph{witness propagation} technique originally introduced by Deng, Mao, Zhong \cite{dmz23} in the context of easier Unbounded Knapsack-type problems. 
Transferring their technique to our 0-1 Knapsack setting needs several additional ideas to be covered later, but the first stage of their technique, \emph{preparing base solutions}, can be applied here easily (except for a small but critical twist on tie-breaking).  

 The idea of witness propagation \cite{dmz23} is to perform DP updates extending from partial solution $\vec x$ to $\vec x + \vec e_w$ only if $x_w >0$. This can lead to improvement if $\|\vec x\|_0$ is small. 
 In order to use this idea, we need to perform a first stage that prepares all the ``base solutions'' $\hat{\vec x} \in \{0,1\}^{\caW}$ that are potentially useful.
 In our case, since optimal solutions vectors $\vec x^*$ satisfy $\|\vec x^*\|_0 \le b_0$, it suffices to prepare those base solutions $\hat{\vec x}$ with  $\|\hat{\vec x}\|_0 \le b_0$. This only requires a DP table of length $O(b_0 \ww)$, which is much shorter than the length needed for a full DP ($O(b_1 \ww)$).

 The property of the algorithm for preparing base solutions is formally summarized in the following lemma.
 We use the following notations for convenience:
for vector $\vec x\in \Z_{\ge 0}^{\caW}$, denote its total profit by $P(\vec x) = \sum_{w\in \caW}P_w(x_w)$, and its total weight by $W(\vec x) = \sum_{w\in \caW} w \cdot  x_w$.
 Let $\min(\vec x,\vec y)$ be the vector whose $i$-th coordinate equals $\min(x_i,y_i)$.

 \begin{lemma}[Preparing base solutions]
    \label{lem:findbase}
 We can compute an array  ${\vec B}[-b_0\ww \dd b_0 \ww]$ of base solutions, where ${\vec B}[i] \in \{0,1\}^{\caW} \cup \{\bot\}$ for all $i$, such that
 \begin{itemize}
    \item for all $i$, $W({\vec B}[i]) = i $ and $\|{\vec B}[i]\|_0 \le b_0$ (unless ${\vec B}[i] = \bot$), and
    \item there exists an optimal solution $\vec y\in \Z_{\ge 0}^{\caW}$ to
   \textsc{KnapsackProximity} (\cref{prob:transformedknapsack}) such that $\min(\vec y, \vec 1) = \vec B[j]$ for $j= W(\min(\vec y,\vec 1))$,
 \end{itemize}
by a deterministic algorithm in $O( b_0 \ww \lvert\caW\rvert)$ time.
 \end{lemma}

\begin{proof}
    The algorithm is basically the standard dynamic program for 0-1 Knapsack on $|\caW|$ items, except that the range of the DP table is restricted to $[-b_0\ww \dd b_0 \ww]$ only.
    We give the pseudocode in \cref{alg:base-sol-knapsack} for completeness.
     The $i$-th entry of a DP table corresponds to a solution vector $\vec x$ with $W(\vec x)= i$ (or an empty solution $\bot$).
 In the $j$-th round, a new item with weight $w_j\in \caW$ (and profit $P_{w_j}(1)$) is used to update the DP table.
 Finally, only those solution vectors with support size $\le b_0$ are returned.
 The returned vectors are in $\{0,1\}^\caW$ since the considered items have distinct weights.
    Clearly, \cref{alg:base-sol-knapsack} performs
    $O( b_0 \ww \lvert\caW\rvert)$
    DP updates, each implemented in $O(1)$ time implicitly by keeping track of back pointers, and finally the solution vectors are reconstructed by following back pointers.  %

\begin{algorithm}
\DontPrintSemicolon
\caption{Prepare base solutions}
\label{alg:base-sol-knapsack}
$\textsc{PrepareBaseSolutions}\colon$\\
\Begin{
Initialize ${\vec {D}}^{(0)}[-b_0\ww \dd b_0 \ww]$ with  ${\vec D}^{(0)}[0] = \vec 0$, ${\vec D}^{(0)}[w]=\bot $ for $w\neq 0$\\
\For{$w_j \in \caW = \{w_1,w_2,\dots,w_{|\caW|}\}$\label{line:baseforloop}}{
    ${\vec D}^{(j)}\gets {\vec D}^{(j-1)}$\\
    \For {$i \in [-b_0\ww , b_0 \ww] \cap [-b_0\ww -w_j, b_0 \ww-w_j ]$} {
\lIf{${\vec D}^{(j)}[i+w_j] = \bot$ or $P({\vec D}^{(j)}[i+w_j])< P({\vec D}^{(j-1)}[i] + \vec{e}_{w_j})$}{${\vec D}^{(j)}[i+w_j] \gets {\vec D}^{(j-1)}[i] + \vec{e}_{w_j}$
}
}
}
\Return{$\vec{B}[i] := \bot $ if  $\|{\vec D}^{(|\caW|)}[i]\|_0> b_0$, otherwise $\vec{B}[i] :=\|{\vec D}^{(|\caW|)}[i]\|_0 $, for all $i$}\label{line:eraselarge}
}
\end{algorithm}

It remains to prove the second property in the lemma statement, which claims one of the computed base solutions is the support of some 
optimal solution to \textsc{KnapsackProximity}.   

Pick an optimal solution vector $\vec x$
to \textsc{KnapsackProximity}, which has support size $\le b_0$ by assumption. Let $\hat {\vec x} = \min(\vec x,\vec 1)$. Since for all  prefixes $J=\{w_1,w_2,\dots,w_j\}\subseteq \caW$ (in the order of the for loop at \cref{line:baseforloop})
  it holds that (here $\vec 1_{J}$ denotes the indicator vector of $J$)
   \[\lvert W(\min(\vec 1_{J},\hat{\vec x})) \rvert \le \left \|\min(\vec 1_{J},\hat{\vec x})\right \|_0 \cdot \ww\le \|\vec x\|_0\cdot  \ww \le b_0 \ww,\]
    that is,  $\min(\vec 1_{J},\hat{\vec x})$ has weight bounded by the range of the DP table, and hence is considered by the DP updates. 
    In particular, $\hat{\vec x}$ is considered by the DP, so the vector $\hat{\vec y}:= {\vec D}^{(|\caW|)}[W(\hat{\vec x})] $ in the DP table cannot be worse: 
    \begin{equation}
        \label{eqn:hatyoptimal}
        P(\hat{\vec y})\ge P(\hat{\vec x}).
    \end{equation}
         Notice $\hat {\vec y}\in \{0,1\}^{\caW}$ and $W(\hat{\vec y}) = W(\hat{\vec x})$.

   Now we define another solution vector $\vec y:=  \vec x- \hat{\vec x} + \hat{\vec y}$; note that $\vec y \ge \hat{\vec y}$.
  Our goal is to prove $\vec y$ is an optimal solution that satisfies the desired property in the lemma statement. 

   First, $\vec y$ has the same weight as $\vec x$,
   \[ W(\vec y) =  W(\vec x)- W(\hat{\vec x}) + W(\hat{\vec y}) = W(\vec x),\] 
    and hence the optimality of $\vec x$ implies
    \begin{equation}
        \label{eqn:yxoptimality}
        P(\vec y)\le P(\vec x).
    \end{equation} 
   
    Then we compare $P(\vec y) -P(\hat{\vec y}) $ and $P(\vec x)-P(\hat{\vec x})$ by looking at each coordinate $w\in \caW$:
    \begin{itemize}
        \item \textbf{Case 1:} $\hat x_w< x_w $. 
        
        By definition $\hat x_w = \min(x_w,1)$, this means $\hat x_w = 1$. Then by $\hat y_w\le 1 = \hat x_w$ and the strict concavity of $P_w$, we have 
\[P_w(y_w)- P_w(\hat y_w) \ge  P_w(x_w)- P_w(\hat x_w),  \]
where strict inequality holds if and only if $\hat y_w <\hat x_w$.

        \item \textbf{Case 2:} $\hat x_w= x_w  $.  Then the equality \[P_w(y_w)   -P_w(\hat y_w)=  P_w(x_w)- P_w(\hat x_w)  =0\] always holds.
    \end{itemize}
    Summarizing the two cases and summing over all $w\in \caW$, we have 
    \begin{equation}
        \label{eqn:yyxx}
     P(\vec y) - P(\hat{\vec y}) \ge P(\vec x) - P(\hat{\vec x}),
    \end{equation}
    where strict inequality holds if and only if there is some $w\in \caW$ such that $\hat y_w < \hat x_w < x_w$.
    
    Combinining \cref{eqn:yyxx} with \cref{eqn:hatyoptimal} and \cref{eqn:yxoptimality}, we know equality is attained in all three of them. In particular, $P(\vec y) = P(\vec x)$, meaning that $\vec y$ is also an optimal solution. 
    
Recall that $\hat{\vec y}\le \vec 1$
 and $\hat{\vec y}\le \vec y$.    We claim $\hat{\vec y} = \min(\vec y, \vec 1)$ must hold. If not, then there must exist $w\in \caW$ such that $\hat y_w = 0$ and $y_w\ge 1$. But this would mean $x_w  - \hat x_w = y_w - \hat y_w = y_w \ge 1$, and hence satisfy the condition $\hat y_w < \hat x_w <x_w$ for strict inequality in \cref{eqn:yyxx}, a contradiction.

Hence we have found an optimal solution $\vec{y}$ whose base solution $\hat{\vec y} = \min(\vec y,\vec 1)$ is in the DP table ${\vec D}^{(|\caW|)}$.
By the assumption that all optimal solutions have support size at most $b_0$, we know $\|\hat{\vec y}\|_0= \|{\vec y}\|_0\le b_0$, so $\hat{\vec y}$ is returned at \cref{line:eraselarge}.
 \end{proof}
Including the support $\hat{\vec y}$ of an optimal solution $\vec y$ as a base solution is crucial for the correctness of witness propagation.
Notice how the tie-breaking conditions imposed in the definition of \cref{prob:transformedknapsack} facilitated our argument involving the support of $\vec y$.  
In \cite{dmz23}, this was treated using the notion of lexical-minimal solutions, which worked well in their unbounded setting but seems not flexible enough to be applied in our setting.

 After preparing the base solutions, it remains to extend them to solution vectors with possibly higher counts (without growing the support size).
 This second stage of witness propagation is the part where our {0-1} Knapsack setting becomes more difficult than the earlier unbounded setting~\cite{dmz23}. 
 Before describing our algorithm for the second stage, we first carefully formulate the task that it is supposed to solve.

We first define the following standard variant of knapsack problem, where we are given a length-$L$ array $q[\cdot ]$, indicating that we can start from a partial solution with profit $q[z]$ and weight $z$, and we want to extend the partial solutions to possibly include additional items, whose weights come from some positive integer set $U$.
\begin{prob}[$\textsc{KnapsackExtend}^+$]
    \label{prob:standarddp}
    Let $U \subseteq [\ww]$. For every $w\in U$,  $Q_w\colon \Z_{\ge 0} \to\Z$ is a strictly concave function with $Q_w(0)=0$ that can be evaluated in constant time.

Given initial profit values $q[0],q[1],\dots,q[L-1] \in \Z$ (and $q[i]=-\infty$ for all $i\notin\{0,1,\dots,L-1\}$), the task is to compute, for each $0\le i\le L-1$, a solution vector $\vec x[i] \in \Z_{\ge 0}^{U}$ that maximizes the final total profit
\begin{equation}
    \label{eqn:tomaximize}
 r[i]:= q\big [z[i]\big ]  + \sum_{w\in U}Q_w(x[i]_w),
\end{equation}
where integer $z[i]$ is uniquely determined by
\begin{equation}
    \label{eqn:subjectto}
z[i] + \sum_{w\in U}w\cdot x[i]_w = i.
\end{equation}
\end{prob}
Note that in \cref{eqn:subjectto}, $z[i] \le i$ must hold, since $w\in U \subseteq [\ww]$ is always positive and $\vec{x}[i]$ is a non-negative vector.
We can analogously define the $\textsc{KnapsackExtend}^-$ problem, which is the same as \cref{prob:standarddp} except that the condition $U \subseteq [\ww]$  is replaced by $U \subseteq -[\ww]$. Note that $\textsc{KnapsackExtend}^-$ can be easily reduced to $\textsc{KnapsackExtend}^+$ by reversing the index range $[0,\dots,L-1]$ and negating all $w\in U$.

Now we define the a weaker version of \cref{prob:standarddp}. This will be the problem that we solve in the second stage.
\begin{prob}[$\textsc{KnapsackExtendWeak}^+$]
    \label{prob:prob3}
    In the same setup as $\textsc{KnapsackExtend}^+$ (\cref{prob:standarddp}),
    we are additionally given sets $S[0],S[1],\dots,S[L-1] \subseteq U$.

The task is to solve \cref{prob:standarddp} with the following relaxation: for each $0\le i\le L-1$, 
\begin{itemize}
    \item 
If \emph{all} maximizers $(z[i],\vec x[i])$ of \cref{eqn:tomaximize} (subject to \cref{eqn:subjectto}) satisfy 
\begin{equation}
    \label{eqn:supportcontain}
\supp(\vec x[i]) \subseteq S\big [z[i]\big ],
\end{equation}
then we are required to correctly output a maximizer for $i$.
\item Otherwise, we are allowed to output a suboptimal solution for $i$.
\end{itemize}
\end{prob}
In \cref{prob:prob3}, when we optimze \cref{eqn:tomaximize}, we can safely restrict attention to solutions that satisfy the support containment condition \cref{eqn:supportcontain}. This captures the idea of witness propagation---to perform DP updates that add some weight $w$ onto $z$, we only care about those $w$'s from the small set $S[z]$.
But the definition of \cref{prob:prob3} is different from (and weaker than) 
 maximizing \cref{eqn:tomaximize} for every $i$ subject to \cref{eqn:supportcontain}. The latter version would make a cleaner definition, but it is a  harder problem which we do not know how to solve.

We also analogously define the $\textsc{KnapsackExtendWeak}^-$ problem where $U\subseteq -[\ww]$.

Our algorithm for \cref{prob:prob3} has the following running time:
\begin{restatable}{theorem}{algolargeb}
    \label{lem:prob3largeb}

   In $\textsc{KnapsackExtendWeak}^+$ (\cref{prob:prob3}),
   suppose there are $m$ distinct sets among $S[0],S[1],\dots,S[L-1]$, and $|S[i]|\le b$ for all $i$. Then $\textsc{KnapsackExtendWeak}^+$ can be solved deterministically in $\tilde O(|U|m+bL)$ time.
\end{restatable}
As a corollary, the same running time holds for the $\textsc{KnapsackExtendWeak}^-$ problem.

Now we show that this algorithm can be applied to the prepared base solutions and correctly solve  \textsc{KnapsackProximity} (\cref{prob:transformedknapsack}), proving \cref{thm:solveknapsackprox}.
\begin{proof}[Proof of \cref{thm:solveknapsackprox} assuming \cref{lem:prob3largeb}]
   To solve the \textsc{KnapsackProximity} problem, we first use \cref{lem:findbase} to find 
 base solutions   ${\vec B}[i] \in \{0,1\}^{\caW} \cup \{\bot\}$ for all $i \in [-b_0\ww , b_0 \ww]$, in $O(b_0 \ww |\caW|)$ time. We extend the indices $i$ to the full range of $i \in [-b_1\ww , b_1 \ww]$, by setting ${\vec B}[i] = \bot$ for all $i\notin [-b_0\ww , b_0 \ww]$.

The plan is to use the $\textsc{KnapsackExtendWeak}^+$ algorithm to extend these base solutions using $P_w(\cdot)$ with positive weights $w\in \caW^+$, and then use the $\textsc{KnapsackExtendWeak}^-$ algorithm to further extend the obtained solutions using $P_w(\cdot)$ with negative weights $w\in \caW^-$. We need to show that some optimal solution survives in the final result.

In more details, we first define a $\textsc{KnapsackExtendWeak}^+$ instance $\big (U,\{Q_w\}_{w\in U}, S[\,] , q[\,]\big )$ where
\begin{align*}
    U&= \caW^+,\\
    Q_w(x) &= P_w(x+1) - P_w(1) & \text{for $w\in \caW^+$ and $x\in \Z_{\ge 0}$,}\\
    q[i] & = P({\vec B}[i])  & \text{for $i\in [-b_1\ww , b_1 \ww]$, }\\
    S[i] & = \supp({\vec B}[i]) \cap \caW^+ &\text{for $i\in [-b_1\ww , b_1 \ww]$, }
\end{align*}
where we shifted (without loss of generality) the index range to $i\in [-b_1\ww , b_1 \ww]$ as opposed to $i\in [0,L-1]$ defined in $\textsc{KnapsackExtendWeak}^+$. Here we assume $P(\bot) = -\infty$ and $\supp(\bot) = \emptyset$.
We use \cref{lem:prob3largeb} to solve 
this $\textsc{KnapsackExtendWeak}^+$ instance, and obtain solutions
\begin{equation}
    \label{eqn:solplus}
 (\vec{x}[i],z[i],r[i])
\end{equation}
for all $i\in [-b_1\ww , b_1 \ww]$ (recall that $z[i]\in [-b_1\ww , b_1 \ww]$ and $r[i]\in \Z$ are uniquely determined by the solution vector $\vec{x}[i]$, as defined in \cref{prob:standarddp}).

Recall from \cref{lem:findbase} that there exists  an optimal solution $\vec y$ to \textsc{KnapsackProximity} such that $\hat{\vec y} = \min(\vec y,\vec 1)$ is one of the base solutions, namely ${\vec B}[W(\hat{\vec y})] = \hat{\vec y}$.
Define vector $\vec y^+ \in \Z_{\ge 0}^{\caW}$ by 
\[ y^+_w = \begin{cases}
    y_w & w\in \caW^+\\
    \min(1,y_w) & w \in \caW^-.
\end{cases}\]
Then we have the following claim:
\begin{claim} Let $i^* = W(\vec y^+)$. Then, the solutions~\cref{eqn:solplus} for the $\textsc{KnapsackExtendWeak}^+$ instance satisfy $r[i^*]  = P(\vec{y}^+)$.

    \label{claim1}
\end{claim}
\begin{proof}[Proof of \cref{claim1}]
    First, note that 
    \[ |i^*| = |W(\vec{y}^+)|\le \|\vec{y}^+\|_1 \cdot \ww \le \|\vec{y}\|_1\cdot \ww \le b_1\ww,\]
    so index $i^*$ is inside the range of the $\textsc{KnapsackExtendWeak}^+$ instance.
   We know that $\vec{y}^+ - \hat{\vec y}$ is a solution vector for $i^*$, with $z$ value (as defined in \cref{eqn:subjectto})
   \[ z = i^* - W(\vec{y}^+ - \hat{\vec y}) = W(\vec y^+) - W(\vec{y}^+ - \hat{\vec y}) = W(\hat{\vec y}),\]
    and objective value (as defined in \cref{eqn:tomaximize})
   \begin{align*}
    r &= q[z] + \sum_{w\in \caW^+} Q_w(y_w-1)  \\
    & = P(\hat {\vec y}) + \sum_{w\in \caW^+} (P_w(y_w-1+1)- P_w(1))\\
    & = P(\vec{y}^+). 
   \end{align*}
Now  suppose for contradiction that   $r[i^*]\neq  P(\vec{y}^+)$. Then by the definition of  $\textsc{KnapsackExtendWeak}^+$, there can only be two possibilities:
\begin{itemize}
    \item Case 1: The solution $(\vec{y}^+ - \hat{\vec y}, W(\hat {\vec y}))$ (with objective value $P(\vec{y}^+)$) is not a maximizer of \cref{eqn:tomaximize} for $i^*$. 
    
    This means there is another solution $(\vec y^\star, z^\star)$ for $i^*$ (where $\vec y^\star \in \Z_{\ge 0}^{\caW^+}$) that achieves an objective value $r^\star$ higher than $P(\vec y^+)$: 
   \begin{equation}
    r^\star = P(\vec B[z^\star]) + \sum_{w\in \caW^+} (P_w(y^\star_w+1)-P_w(1)) > P(\vec{y}^+).
    \label{eqn:rstar}
   \end{equation} 
Define vector $\bar{\vec y}^+:= \vec {y^\star} +\vec B[z^\star]$. It has total weight 
\[ W(\bar{\vec y}^+) = W(\vec B[z^\star]) + W(\vec {y^\star})  = i^*,\]
and total profit 
\[ P(\bar{\vec y}^+) \ge P(\vec B[z^\star]) + \sum_{w\in \caW^+} (P_w(y^\star_w+1)-P_w(1)) = r^\star > P(\vec{y}^+),\]
where the first inequality follows from 
 the concavity of $P_w(\cdot)$ and  $\vec B[z^\star]  \le \vec 1$.

Hence we have obtained a vector $\bar{\vec y}^+$ where $W(\bar{\vec y}^+) = W(\vec y^+)$ and $\bar{y}^+_w \le 1$ for all $w\in \caW^-$, but $P(\bar{\vec y}^+) > P({\vec y}^+)$. From this $\bar{\vec y}^+$ we can use a simple exchange argument to find a strictly better solution to $\textsc{KnapsackProximity}$ than $\vec{y}$, and get a contradiction.
Specifically,    define the following solution vector \[\vec y'':=  \bar{\vec y}^+ + (\vec y - \vec y^+) \in \Z_{\ge 0}^\caW\]
    for the \textsc{KnapsackProximity} instance, which has total weight
    \[ W(\vec y'')= W( \bar{\vec y}^+) + W(\vec y) - W(\vec y^+) = W(\vec y),\]
    and total profit 
    \begin{align*}
        P(\vec y'')&\ge  P( \bar{\vec y}^+) +  \sum_{w\in \caW^-} (P_w(y_w)-P_w(1))\\
        & = P( \bar{\vec y}^+)+  P(\vec y) - P(\vec y^+)\\
        & > P(\vec y).
    \end{align*}
    where in the first inequality we used $\supp(\vec y - \vec y^+) \subseteq \caW^-$, and $\bar{y}^+_w \le 1$ for all $w\in \caW^-$.
    \item 
Case 2: The solution $(\vec{y}^+ - \hat{\vec y}, W(\hat {\vec y}))$ (with objective value $P(\vec{y}^+)$) is a maximizer of \cref{eqn:tomaximize} for $i^*$, but there is also another maximizer  $(\vec y^\star, z^\star)$ for $i^*$ with the same objective value $r^\star = P(\vec y^+)$  that violates the support containment condition (\cref{eqn:supportcontain}). Hence $\supp(\vec y^\star)\nsubseteq S[z^\star] = \supp(\vec B[z^\star])\cap \caW^+$.

Again define vector $\bar{\vec y}^+:= \vec {y^\star} +\vec B[z^\star]$. It has total weight 
\[ W(\bar{\vec y}^+) = W(\vec B[z^\star]) + W(\vec {y^\star})  = i^*,\]
and total profit 
\begin{align}
    P(\bar{\vec y}^+) &= P(\vec B[z^\star]) + \sum_{w\in \caW^+ \cap \supp(\vec B[z^\star])}(P_w(y^\star_w+1)-P_w(1)) + \sum_{w\in \caW^+ \setminus \supp(\vec B[z^\star])}P_w(y^\star_w)\nonumber \\
    & > P(\vec B[z^\star]) + \sum_{w\in \caW^+ \cap \supp(\vec B[z^\star])}(P_w(y^\star_w+1)-P_w(1)) + \sum_{w\in \caW^+ \setminus \supp(\vec B[z^\star])}(P_w(y^\star_w+1) - P_w(1))\label{eqn:strict} \\
    & = r^\star = P(\vec y^+), \nonumber
\end{align}
where the strict inequality follows from the existence of some $w\in \supp(\vec y^\star)\setminus \supp(\vec B[z^*])$, which satisfies $P_w(y^*_w)>P_w(y^*_w+1) - P_w(1)$ due to $y_w\ge 1$ and the strict concavity of $P_w(\cdot )$.

Hence we have obtained a vector $\bar{\vec y}^+$ where $W(\bar{\vec y}^+) = W(\vec y^+)$ and $\bar{y}^+_w \le 1$ for all $w\in \caW^-$, but $P(\bar{\vec y}^+) > P({\vec y}^+)$. Then we can derive a contradiction in the same way as in the previous case. \qedhere
\end{itemize}
\end{proof}

We have established that \cref{eqn:solplus} satisfies $r[i^*]  = P(\vec{y}^+)$.
Now, we define a $\textsc{KnapsackExtendWeak}^-$ instance
$\big (U',\{Q_w\}_{w\in U'}, S'[\,] , q'[\,]\big )$ where
\begin{align*}
    U'&= \caW^-,\\
    Q_w(x) &= P_w(x+1) - P_w(1) & \text{for $w\in \caW^-$ and $x\in \Z_{\ge 0}$,}\\
    q'[i] & = r[i]  & \text{for $i\in [-b_1\ww , b_1 \ww]$, }\\
    S'[i] & = \supp({\vec B}\big [z[i]\big ]) \cap \caW^- &\text{for $i\in [-b_1\ww , b_1 \ww]$, }
\end{align*}
where we again assume (without loss of generality) the index range is $i\in [-b_1\ww , b_1 \ww]$ instead of $i\in [0,L-1]$.  We use \cref{lem:prob3largeb} to solve this $\textsc{KnapsackExtendWeak}^-$ instance, and obtain solutions
\begin{equation}
 (\vec{x}'[i],z'[i],r'[i])
\end{equation}
for all $i\in [-b_1\ww , b_1 \ww]$. We now claim that $r'[W(\vec y)]=P(\vec y)$, which means we have successfully found the optimal solution. The proof of this claim is similar to the previous claim using exchange arguments, and here we give a proof sketch.

First notice $|W(\vec y)|\le b_1 \ww $ so that it falls into the index range of the instance. Then we know $\vec y - \vec y^+$ is a solution vector for index $W(\vec y)$ in the $\textsc{KnapsackExtendWeak}^-$ instance, with  objective value $P(\vec y)$ (where we used the fact that $q'[W(\vec y^+)] = r[W(\vec y^+)]  = P(\vec{y}^+)$).
Suppose for contradiction that $r'[W(\vec y)] \neq P(\vec y)$, then similar to the Case 2 in  the  proof of \cref{claim1}, we know there is another maximizer that violates the support containment condition, and from there we can use an argument similar to \cref{eqn:strict} based on strict concavity, and obtain a solution to $\textsc{KnapsackProximity}$ with strictly higher total profit than $\vec y$, contradicting the optimality of $\vec y$.

Hence, we have shown how to use two applications of \cref{lem:prob3largeb} to solve $\textsc{KnapsackProximity}$. It remains to analyze the time complexity.
Since there are only $O(b_0\ww)$ many distinct base solutions $\vec B[i]$, the time complexity for both applications of \cref{lem:prob3largeb} is $\tilde O(|U|m+bL) = \tilde O(|\caW|\cdot b_0\ww + b_0\cdot b_1\ww)$. More precisely (see the end of the proof of \cref{lem:prob3largeb} in \cref{subsec:color-coding}), it is  $O(b_0\ww ( |\caW|  \log (b_0\ww ) +  b_1) \log^2 (b_0\ww))$. (The time complexity $O(b_0 \ww |\caW|)$ of applying \cref{lem:findbase} is dominated.)
\end{proof}

\subsection{An algorithm for singleton sets $S[i]$}
\label{subsec:singleton}

We start to describe our algorithms for
$\textsc{KnapsackExtendWeak}^+$ (\cref{lem:prob3largeb}).
Our most interesting building block is an algorithm for solving $\textsc{KnapsackExtendWeak}^+$ in the case where the given sets satisfy $|S[i]|\le 1$ for all $0\le i\le L-1$.

\begin{algorithm}
\DontPrintSemicolon
\caption{Solving $\textsc{KnapsackExtendWeak}^+$ for singleton sets}
\label{alg:knapsack-batch-update}
    Given $S[0\dd L-1]$ where $S[i] \subseteq [\ww]$ and $|S[i]|\le 1$ for all $i$\\
    Given values $q[0\dd L-1]$ where $q[i] \in \Z$\\
$\textsc{SMAWKAndScan}(q[0\dd L-1],S[0\dd L-1])$\\
\Begin{
    \tcc{Stage 1: use SMAWK to find all candidate updates $q[j] + Q_w(\cdot)$ where $w\in S[j]$, expressed as APs}
    Initialize $\caS \gets \emptyset$\\
    \For{$w\in [\ww]$ and $c\in \{0,1,\dots,w-1\}$\label{line:forloop}} {
        $J := \{j : w\in S[j]\text{ and } j\equiv c\pmod{w}, 0\le j\le L-1\}$ \label{line:defnJ}\\
        $I := \{i :  i\equiv c\pmod{w}, 0\le i\le L-1\}$\\
        Run SMAWK (\cref{thm:smawk}) on matrix $A_{I\times J}$ defined as $A[i,j]:= q[j] + Q_w\left (\frac{i-j}{w} \right )$. \label{line:smawk}\\
        \For{$j\in J$}{
        Suppose SMAWK returned the AP  $P_j = \{c+k w, c+(k+1)w,\dots,c+\ell w\}\subseteq I$, such that for every $i\in P_j$, $j = \arg \max_{j'\in J}A[i,j']$\\
         Insert $(j;c,w,k,\ell)$ into $\caS$\\ %
        }
    }
    \tcc{Stage 2: combine all candidate updates by a linear scan from left to right, extending winning APs and discarding losing APs}
    Initialize empty buckets $B[0],B[1],\dots,B[L-1]$\\
    \For{$(j;c,w,k,\ell)\in \caS$}{
        \If{$c+kw>j$}{
        Insert $(j;c,w,k,\ell)$ into bucket $B[c+k w]$ \tcp*[r]{insert to the bucket at the beginning of the AP}
        }\Else{
            \tcp*[l]{a technical corner case $c+kw = j$: separately insert the first element}
        Insert $(j;c,w,k,k)$ into bucket $B[c+k w]$ \label{line:corner}\\
        \lIf{$k+1\le \ell$}{Insert $(j;c,w,k+1,\ell)$ into bucket $B[c+(k+1) w]$}
        }
    }
\For{$i \gets 0,1,\dots,L-1$ \label{line:scan}}{
    $r[i] \gets q[i], z[i] \gets i, \vec{x}[i] \gets \vec 0$.\label{line:trivial} \tcp*[r]{the trivial solution for $i$}
    \If{$B[i] \neq \emptyset$}{
    Pick $(j;c,w,k,\ell)\in B[i]$ that maximizes $q[j] + Q_w\left (\frac{i-j}{w} \right )$ \label{line:checkbucket}\\
    \If{$q[j] + Q_w\left (\frac{i-j}{w} \right ) > r[i]$}{
     $r[i] \gets q[j] + Q_w\left (\frac{i-j}{w} \right ), z[i] \gets j, \vec{x}[i] \gets \frac{i-j}{w}\vec e_w$. \label{line:updatei}\tcp*[r]{solution for $i$}
    }
     \If{$i+w\le c+\ell w$}{\label{line:insertanother} Insert $(j;c,w,k,\ell)$ into bucket $B[i+w]$ \tcp*[r]{extend this winning AP by one step, and all other APs in the bucket $B[i]$ are discarded} }
    }
}
\Return{$(\vec x[0\dd L-1],  z[0\dd L-1], r[0\dd L-1])$}
}
\end{algorithm}

\begin{lemma}
    \label{lem:singleton}
    $\textsc{KnapsackExtendWeak}^+$ where $|S[i]|\le 1$ for all $i$ can be solved in $O(L \log L)$ time.
\end{lemma}

Our algorithm for \cref{lem:singleton} is given in \cref{alg:knapsack-batch-update}.  It contains two stages:
\begin{itemize}
    \item In the first stage, we enumerate $w\in [\ww]$ and some $c$ modulo $w$, and collect all sets $S[j]$ containing $w$ with indices $j$ congruent to $c$ modulo $w$. Then we try to extend from these indices $j$ by adding integer multiples of $w$ (which does not interfere with other congruence classes modulo $w$). 
This modulo $w$ idea is standard and has been used in many previous knapsack algorithms, e.g., \cite{DBLP:journals/jco/KellererP04,DBLP:conf/soda/Chan18a,DBLP:conf/icalp/AxiotisT19,icalp21}. This idea is usually used together with concavity arguments. Here we also do so: 
we use SMAWK algorithm \cite{smawk} to compute, for every $i$ in this congruence class, the $j$ that maximizes $q[j]+Q_w(\frac{i-j}{w})$. But in our scenario with small sets $S[\cdot ]$, the number of available $j$'s is usually sublinear in the array length $L$,  so we need to let SMAWK return a compact output representation, described as several arithmetic progressions (APs) with difference $w$, where each AP contains the $i$'s that have a particular $j$ as maximizer.

\item The second stage is to combine all the APs returned by the SMAWK algorithm, and update them onto a single DP array. Ideally, we would like to the (entry-wise) maximum over all the APs. Unfortunately, the total length of these APs could be much larger than the array length $L$, which would prevent us from getting an $\tilde O(L)$ time algorithm.  The idea here is to crucially use the weakening in the definition of $\textsc{KnapsackExtendWeak}^+$, so that we can skip a lot of computation. We perform a linear scan from left to right, and along the way we discard many APs that cannot contribute to any useful answers. In this way we can get the time complexity down to near-linear.
\end{itemize}

\begin{proof}[Proof of \cref{lem:singleton}]
The algorithm is given in \cref{alg:knapsack-batch-update}.  We first analyze the time complexities of its two stages.
\begin{itemize}
    \item 
 The first stage contains a for loop over $w\in \ww$ and $c\in \{0,1\dots,w-1\}$ (\cref{line:forloop}). 
 Note that we only need to execute the loop iterations such that the index set $J := \{j : w\in S[j]\text{ and } j\equiv c\pmod{w}, 0\le j\le L-1\}$ (defined at \cref{line:defnJ}) is non-empty. Since $S[j]\le 1$ for all $j$, these sets $J$ over all $(w,c)$ form disjoint subsets of $\{0,1,\dots,L-1\}$, and can be prepared efficiently. Then, for each of these sets $J$, at \cref{line:smawk} we run a SMAWK algorithm with compact output (\cref{thm:smawk}) in $O(|J|\log L)$ time, and then insert $|J|$ APs into $\caS$. The total running time of this stage is thus $O(L\log L)$.
 Set $\caS$ contains at most $L$ APs (each AP only takes $O(1)$ words to describe).
 \item In the second stage, we initialize $L$ buckets $B[0\dd L-1]$, and first insert each AP from $\caS$ into a bucket (or two buckets, in the corner case at \cref{line:corner}). Then we do a scan $i\gets 0,1,\dots,L-1$ (\cref{line:scan}), where for each $i$ we examine all APs in the bucket $B[i]$ at \cref{line:checkbucket}, and then insert at most one winning AP to another bucket (\cref{line:insertanother}). Hence, in total we only ever inserted at most $|\caS|+L \le 2L$ APs. So the second stage takes $O(L)$ overall time.
\end{itemize}
Hence the time complexity of \cref{alg:knapsack-batch-update} is $O(L\log L)$. It remains to prove that its return values
$(\vec x[i],  z[i], r[i])$ correctly solve $\textsc{KnapsackExtendWeak}^+$. Fix any $i\in \{0,1,\dots,L-1\}$, and let $(\vec x^*[i],  z^*[i], r^*[i])$ be an maximizer of \cref{eqn:tomaximize} (subject to \cref{eqn:subjectto}). If $\lvert\supp(\vec x^*[i])\rvert\ge 2$, then it clearly violates the support containment condition \cref{eqn:supportcontain} because $S[z^*[i]]\le 1$, and hence we are not required to correctly solve for $i$. If $\lvert\supp(\vec x^*[i])\rvert =  0$, then it is the trivial solution, which cannot be better than our solution, due to \cref{line:trivial}. Hence, it remains to consider the $\lvert \supp(\vec x^*[i])\rvert = 1$ case. Let $\supp(\vec x^*[i]) = \{ w^*\}$, and we can assume $S\big [z^*[i]\big ] = \{w^*\}$ (otherwise, it fails the support containment condition \cref{eqn:supportcontain} and hence we are not required to correctly solve $i$).

In the for loop of the first stage where $w=w^*$ and $c = i\bmod w$, we have $z^*[i] \in J$ and $i\in I$.
The SMAWK matrix $A_{I\times J}$ encodes the objective values of extending from $j$ by adding multiples of $w^*$. In particular $A[i,z^*[i]]$ equals our optimal objective $r^*[i]= q\big [z^*[i]\big ] + Q_w\left ( \frac{i-z^*[i]}{w^*}\right )$. So SMAWK correctly returns an AP $P_{z^*[i]} = \{c+k w^*, c+(k+1)w^*,\dots,c+\ell w^*\}$ that contains $i$ (unless there is  a tie $A[i,z^*[i]] = A[i,j]$ for some other $j\in J$ and $i$ ends up in the AP $P_j$, but in this case we can redefine $z^*[i]\gets j$ from now on).

In the second stage, each AP in $\caS$ starts in the bucket labeled by the leftmost element of this AP, and during the left-to-right linear scan this AP may win over others in its current bucket (at \cref{line:checkbucket}) and gets advanced to the bucket corresponding to its next element in the AP (at \cref{line:insertanother}), or it may lose at \cref{line:checkbucket} and be discarded.  
Our goal is to show that the AP $P_{z^*[i]}$ can survive the competitions and arrive in bucket $B[i]$, so that it can successfully update the answer for $i$  at \cref{line:updatei}.
Suppose for contradiction that it lost to some other AP $P'_{j'}$ when they were both in bucket $B[i_0]$ (for some $i_0<i$). (By the way we handled the corner case at \cref{line:corner}, here we can assume $i_0\ge z^*[i]+1$.) Suppose this AP $P'_{j'}$ has common difference $w'$, and corresponds to the objective value $q[j'] + Q_{w'}\left ( \frac{i'-j'}{w'}\right )$ for $i'\in P'_{j'}$. 
 Now we consider an alternative solution for index $i$ defined as 
\[ \big(\tfrac{i_0-j'}{w'}\vec e_{w'} + \tfrac{i-i_0}{w^*}\vec e_{w^*}\,,\, j'\big ), \]
which has objective value
\begin{align*}
 &    q[j'] + Q_{w'}\left ( \frac{i_0-j'}{w'}\right ) + Q_{w^*}\left ( \frac{i-i_0}{w^*}\right )\\
 \ge \ & q\left [z^*[i]\right ] + Q_{w^*}\left ( \frac{i_0-z^*[i]}{w^*}\right ) + Q_{w^*}\left ( \frac{i-i_0}{w^*}\right ) \tag{since $P_{j'}$ wins over $P_{z^*[i]}$ in bucket $B[i_0]$}\\
 > \ & q\left [z^*[i]\right ] + Q_{w^*}\left ( \frac{i_0-z^*[i]}{w^*} + \frac{i-i_0}{w^*}\right ) + 0 \tag{by $i>i_0>z^*[i]$ and strict concavity of $Q_{w^*}$}\\
 = \ & r^*[i],
\end{align*}
which contradicts the assumption that $r^*[i]$ is the optimal objective value for index $i$.
Hence, we have shown that the AP $P_{z^*[i]}$ can arrive in bucket $B[i]$. This finishes the proof that \cref{alg:knapsack-batch-update} correctly solves $\textsc{KnapsackExtendWeak}^+$ for index $i$.
\end{proof}

\subsection{Helper lemmas for $\textsc{KnapsackExtendWeak}^+$}
\label{subsec:help}
In this section, we show several helper lemmas for the 
$\textsc{KnapsackExtendWeak}^+$ problem. 
To get some intuition, we first observe that its unweakened version, $\textsc{KnapsackExtend}^+$, is a standard dynamic programming problem which should obey some kind of composition rule: namely, if we apply a $\textsc{KnapsackExtend}^+$ algorithm to extend a partial DP array $q[\cdot]$ with items from $U_1\cup U_2$ (for some disjoint $U_1$ and $U_2$), it should have the same effect as first extending $q[\cdot]$ with $U_1$, obtaining an intermediate DP array, and then extending this intermediate array with $U_2$.

The main goal of this section is formulate and prove analogous composition properties for the $\textsc{KnapsackExtendWeak}^+$ problem. These properties will be useful for our decomposition-based algorithms to be described later in \cref{subsec:color-coding}.

Using the notations from \cref{prob:prob3}, we denote an instance of $\textsc{KnapsackExtendWeak}^+$ as 
\[K = \big (U,\{Q_w\}_{w\in U}, S[0\dd L-1] , q[0\dd L-1]\big ),\]
where $S[i] \subseteq U$ for all $i$.  And we denote a solution to $K$ as 
\[ Y = (\vec x[0\dd L-1], z[0\dd L-1], r[0\dd L-1]).\]
where $z[i]$ is the starting index uniquely determined by $i$ and $\vec x$ by \cref{eqn:subjectto}, and 
\[ r[i]:= q\big [z[i]\big ] + \sum_{w\in V}Q_w(x[i]_w)\]
is the objective value.

Now we define several operations involving the instance $K$. In the following we omit the array index range $[0\dd L-1]$ for brevity.
\begin{definition}[Restriction]
The \emph{restriction} of instance $K$ to a set $V\subseteq U$ is defined as the instance
\[ K\lvert_{V}:= \big (V,\{Q_w\}_{w\in V}, S_V[\, ] , q[\, ]\big )\]
where $S_V[i]:= S[i] \cap V$.
\end{definition}

\begin{definition}[Updating]
Suppose $Y_V = (\vec x[\, ], z [\, ], r[\, ])$ is a solution to $K\lvert_{V}$, then we define the following updated instance
\[ K^{(V\gets Y_V)} := (U\setminus V, \{Q_w\}_{w\in U\setminus V}, S'[\,],q'[\,])\]
where 
\begin{equation}
    \label{eqn:nextsets}
 S'[i]:= S\big [z[i]\big] \setminus V,
\end{equation}
and
\begin{equation*}
 q'[i]:= r[i].
\end{equation*}
\end{definition}

\begin{definition}[Composition]
    Let $V,V'\subseteq U, V\cap V'= \emptyset$.
Suppose $Y_V = (\vec x[\, ], z [\, ], r[\, ])$ is a solution to $K\lvert_{V}$, and  $Y_{V'}= ({\vec x}'[\, ], z' [\, ], r'[\, ])$ is a solution to $K^{(V\gets Y_V)}\lvert_{V'}$. 
We define the following composition of solutions,
\[ Y_{V'} \circ Y_V := ({\vec x}''[\, ], z'' [\, ], r'[\, ]), \]
where 
\[ z''[i]:= z\big [z'[i]\big],\]
\[ {\vec x}''[i]:= {\vec x}'[i] + {\vec x}\big [z'[i]\big ].\]
Note that $\circ$ is associative.
\end{definition}

Now we are ready to state the composition lemma for $\textsc{KnapsackExtendWeak}^+$.
\begin{lemma}[Composition lemma for $\textsc{KnapsackExtendWeak}^+$]
    \label{lem:composeweak}
    Let $V,V'\subseteq U, V\cap V'= \emptyset$.
    If $Y_V$  correctly solves $K\lvert_V$, and $Y_{V'}$ correctly solves $K^{(V\gets Y_V)}\lvert_{V'}$, then $Y_{V'} \circ Y_V$ correctly solves $K\lvert_{V\cup V'}$.
\end{lemma}
\begin{proof}
The proof strategy here is reminiscent of the proof of \cref{thm:solveknapsackprox} in \cref{subsec:knapsack-base}.

Recall $K = \big (U,\{Q_w\}_{w\in U}, S[\,] , q[\,]\big )$.
Denote the solutions $Y_V = (\vec x[\, ], z [\, ], r[\, ])$, $Y_{V'}= ({\vec x}'[\, ], z' [\, ], r'[\, ])$, and $Y_{V'} \circ Y_V := ({\vec x}''[\, ], z'' [\, ], r'[\, ])$.

Suppose for contradiction that for some $i$, ${\vec x}''[i]$ is incorrect for the instance $K\lvert_{V\cup V'}$.
By definition of $\textsc{KnapsackExtendWeak}^+$, this means that all maximizers $(\bar{\vec x}, \bar{z})$ (where $\bar{\vec x}\in \Z_{\ge 0}^{V\cup V'}$) of the objective value
 \[ \bar r = q[\bar z]  + \sum_{w\in V\cup V'}Q_w(\bar x_w)\]
 (subject to $\bar z + \sum_{w\in V\cup V'}w\cdot \bar x_w = i$)
should satisfy the support containment condition, 
\[\supp(\bar{\vec x}) \subseteq S[\bar z]. \]
Take any such maximizer $(\bar{\vec x}, \bar{z})$, and write $\bar {\vec x} = \bar {\vec x}_{V} + \bar {\vec x}_{V'}$ such that $\supp(\bar {\vec x}_{V})\subseteq V$ and $\supp(\bar {\vec x}_{V'})\subseteq V'$. Let 
\[ i_V:= \bar z + \sum_{w\in V}w\cdot \bar x_w\]
and
\[r_{V}:= q[\bar z] + \sum_{w\in V}Q_w(\bar x_w).\]
Note that in instance $Y_V$, $(\bar{\vec x}_{V}, \bar z, r_{V})$ should be a valid solution for $i_V$ with objective value $r_V$. We now compare it with $r[i_V]$ from the correct solution $Y_V$ to the instance $K\lvert_{V}$, and we claim that $r[i_V]=r_V$ must hold. Otherwise, by the  definition of $\textsc{KnapsackExtendedWeak}^+$ instance $K\lvert_{V}$,  there can only be two possibilities:
\begin{itemize}
    \item $(\bar{\vec x}_{V}, \bar z, r_{V})$ is not a maximizer solution for index $i_V$ in $K\lvert_{V}$.
    
    This means there is some solution $(\vec x^\star,  z^\star, r^\star)$ for index $i_V$ in $K\lvert_{V}$ that achieves a higher objective $r^\star > r_V$. 
    We will use an exchange argument to derive contradiction:  Consider  the solution $(\vec x^\star + \bar {\vec x}_{V'}, z^\star)$ to the instance $K_{V\cup V'}$. It has equal total weight $z^\star + \sum_{w\in V} x^\star_v + \sum_{w\in V'}(\bar {x}_V)_w = i_V +  \sum_{w\in V'}(\bar {x}_V)_w = i$, but with a higher objective value 
    $ r^\star + \sum_{w\in V'}Q_w(\bar x_w) > r_V + \sum_{w\in V'}Q_w(\bar x_w) = \bar r$, contradicting to the assumption that $(\bar{\vec x}, \bar{z})$ is a maximizer for $i$ in instance $K_{V\cup V'}$.

    \item $(\bar{\vec x}_{V}, \bar z, r_{V})$ is a maximizer solution for index $i_V$ in $K\lvert_{V}$, but there is another maximizer solution $(\vec x^\star,  z^\star, r^\star)$ for index $i_V$ in $K\lvert_{V}$ that does not satisfy the support containment condition $\supp(\vec x^\star)\subseteq S[z^\star]$ for instance $K_{V}$.
    
    In this case, again consider the solution $(\vec x^\star + \bar {\vec x}_{V'}, z^\star)$ to the instance $K_{V\cup V'}$. This time it has  the same objective value $\bar r$, so it is a maximizer for index $i$ in the instance $K_{V\cup V'}$. However, since $\supp(\vec x^\star)\nsubseteq S[z^\star]$, we know $\supp(\vec x^\star+ \bar {\vec x}_{V'})\nsubseteq S[z^\star]$, and hence it violates the support containment condition in $K_{V\cup V'}$, contradicting our assumption that all maximizers to index $i$ satisfy the support containment condition.
\end{itemize}
Hence we have established that $r[i_V]=r_V$ must hold.

Now we look at the second instance, $K^{(V\gets Y_V)}\lvert_{V'} = \big (V',\{Q_w\}_{w\in V'}, S'[\,] , q'[\,]\big )$. Note that we have $q'[i_V]=r[i_V]=r_V$ by definition. Hence, $(\bar{\vec x}_{V'},i_V,\bar r)$ should be a valid solution for index $i$, with objective value
\[ q'[i_v] + \sum_{w\in V'}Q_w(\bar x_w) = r_V + \sum_{w\in V'}Q_w(\bar x_w) = \bar r.\]
Now we claim that $r'[i] = \bar r$ (where $r'[\cdot]$ denotes objective values realized by the solution $Y_{V'}$) must hold. Otherwise, by the definition of 
$\textsc{KnapsackExtendedWeak}^+$ instance $K^{(V\gets Y_V)}\lvert_{V'}$,  there can only be two possibilities:
\begin{itemize}
    \item $(\bar{\vec x}_{V'},i_V,\bar r)$ is not a maximizer solution for index $i$ in $K^{(V\gets Y_V)}\lvert_{V'}$.

    This means there is some solution for index $i$ in $K^{(V\gets Y_V)}\lvert_{V'}$ that achieves a higher objective. By a similar argument as above, this would contradict the assumption that $(\bar{\vec x}, \bar{z})$ is a maximizer for $i$ in instance $K_{V\cup V'}$.

    \item $(\bar{\vec x}_{V'},i_V,\bar r)$ is a maximizer solution for index $i$ in $K^{(V\gets Y_V)}\lvert_{V'}$, but there is another maximizer solution that does not satisfy the support containment condition.

    Again, by a similar argument, this would contradict our assumption that all maximizers to index $i$ in instance $K|_{V\cup V'}$ satisfy the support containment condition.
\end{itemize}
Hence, we must have $r'[i] = \bar r$. This means that we indeed have found a maximizer to index $i$ after we compose the solutions $Y_V$ and $Y_{V'}$.  So our solution for $i$ is actually correct for the instance $K\lvert_{V\cup V'}$, contradicting our assumption.

Hence, we have established that $Y_{V'} \circ Y_V$ correctly solves $K\lvert_{V\cup V'}$.
\end{proof}

In addition to decompose an instance by partitioning the set $U$, we also need another way to decompose an instance, which in some sense allows us to partition the array indices $[0\dd L-1]$. First, we define the entry-wise maximum of two instances.
\begin{definition}[Entry-wise maximum]
Given $\textsc{KnapsackExtendWeak}^+$ instances $K = \big (U,$ $\{Q_w\}_{w\in U}, S[\, ] , q[\,]\big )$ and 
$K' = \big (U,\{Q_w\}_{w\in U}, S'[\, ] , q'[\,]\big )$,
we define the following instance,
\[(U,\{Q_w\}_{w\in U}, S''[\, ] , q''[\,]\big ), \]
where 
\[ (S''[i],q''[i]):= \begin{cases}
    (S[i],q[i]) & q[i]>q'[i],\\
    (S'[i],q'[i]) & q[i]<q'[i],\\
    (S[i] \cap S'[i],q[i]) & q[i]=q'[i],\\
\end{cases}\]
and denote this instance by $\max(K,K')$.
We naturally extend this definition to the entry-wise maximum of possibly more than two instances.

We also define the entry-wise maximum of two solutions
$Y = (\vec x[\, ], z [\, ], r[\, ]), Y' = (\vec x'[\, ], z' [\, ], r'[\, ])$, by $\max(Y,Y') = (\vec x''[\, ], z'' [\, ], r''[\, ])$, where 
\[ (\vec x''[i],z''[i]):= \begin{cases}
   (\vec x[i],z[i]) & r[i]>r'[i],\\
   (\vec x'[i],z'[i]) & \text{otherwise,}\\
\end{cases}\]
and objective $r''[i]$ can be uniquely determined from $(\vec x''[i],z''[i])$. Note that $r''[i] \ge \max(r[i],r'[i])$ obviously holds.
\end{definition}

Naturally, we have the following lemma for $\textsc{KnapsackExtendWeak}^+$.
\begin{lemma}[Entry-wise maximum lemma]
    \label{lemma:instancemax}
    If $Y$ correctly solves $K$, and $Y'$ correctly solves $K'$, then $\max(Y,Y')$ correctly solves $\max(K,K')$.
\end{lemma}
\begin{proof}
Denote $K = \big (U,\{Q_w\}_{w\in U}, S[\,] , q[\,]\big )$, $K' = \big (U,\{Q_w\}_{w\in U}, S'[\,] , q'[\,]\big )$, $Y = (\vec x[\, ], z [\, ], r[\, ]), Y' = (\vec x'[\, ], z' [\, ], r'[\, ])$. Let $\max(K,K') = \big (U,\{Q_w\}_{w\in U}, S''[\,] , q''[\,]\big )$ and $\max (Y,Y') = (\vec x''[\, ], z'' [\, ], r''[\, ])$.
    
Suppose for contradiction that for some $i$, ${\vec x}''[i]$ is incorrect for the instance $\max(K,K')$.
By definition of $\textsc{KnapsackExtendWeak}^+$, this means that all maximizers $(\bar{\vec x}, \bar{z})$ (where $\bar{\vec x}\in \Z_{\ge 0}^{U}$) of the objective value
 \[ \bar r = q''[\bar z]  + \sum_{w\in U}Q_w(\bar x_w)\]
 (subject to $\bar z + \sum_{w\in U}w\cdot \bar x_w = i$)
should satisfy the support containment condition, 
\[\supp(\bar{\vec x}) \subseteq S''[\bar z]. \]
Take any such maximizer $(\bar{\vec x}, \bar{z})$.
Since $q''[\bar z] = \max(q[\bar z], q'[\bar z])$, without loss of generality we assume $q''[\bar z] = q[\bar z]$. Then, in instance $K$, $(\bar{\vec x},\bar z)$ is a valid solution for index $i$, achieving objective value $q[\bar z] + \sum_{w\in U}Q_w(\bar x_w) = \bar r$. If in solution $Y$, the found objective $r[i]$ satisfies $r[i]\ge \bar r$, then in the entry-wise maximum solution we would have $r''[i]\ge \max(r[i],r'[i])\ge \bar r$, which contradicts the assumption that $\vec x''[i]$ is incorrect. Hence, $r[i]<\bar r$ holds.  Then, since $Y$ is a correct solution to instance $K$, we know from the definition of $\textsc{KnapsackExtendedWeak}^+$ instance $K$ that there can only be two possibilities:
\begin{itemize}
    \item $(\bar{\vec x},\bar z)$ is not a maximizer solution for index $i$ in $K$.

    This would immediately mean that the actual maximum objective for index $i$ in instance $\max(K,K')$ is also greater than $\bar r$, a contradiction.

    \item $(\bar{\vec x},\bar z)$ is a maximizer solution for index $i$ in $K$, but there is another maximizer solution $(\vec x^\star, z^\star)$ for index $i$ that does not satisfy the support containment condition $\supp(\vec x^\star)\subseteq S[z^\star]$ for instance $K$.

    Then, there are two cases:
    \begin{itemize}
        \item $q[z^*]< q'[z^*]$.

        Then, in instance $\max(K,K')$, $(\vec x^*,z^*)$  actually achieves a higher objective $\bar r+ q'[z^*]- q[z^*]$, a contradiction.
        \item $q[z^*]\ge q'[z^*]$. Then, in instance $\max(K,K')$, by definition we have $S''[z^*] \subseteq S[z^*]$. Hence, $(\vec x^*,z^*)$ is a maximizer solution for index $i$ in instance $\max(K,K')$ that does not satisfy the support containment condition $\supp(\vec x^\star) \subseteq S''[z^\star]$, a contradiction.
    \end{itemize}
\end{itemize}
Hence,  we have reached contradictions in all cases. This means $\max(Y,Y')$ is a correct solution to $\max(K,K')$.
\end{proof}

 \subsection{Color-coding}
 \label{subsec:color-coding}

 In \cref{subsec:singleton}, we solved $\textsc{KnapsackExtendWeak}^+$ (\cref{prob:prob3}) where the maximum set size $b$ is at most $1$.
In this section, we extend to larger size $b$ by using the color-coding technique \cite{DBLP:journals/jacm/AlonYZ95} to isolate the elements in the sets, and using the helper lemmas from \cref{subsec:help} to combine the solutions for different color classes.  
 A two-level color-coding approach was previously used in the near-linear time randomized subset sum algorithm of Bringmann \cite{DBLP:conf/soda/Bringmann17}, and our approach here is analogous.
One small difference is that, in our case the sets to be isolated are already given to us as input, so we can derandomize the color-coding technique (whereas derandomizing Bringmann's subset sum algorithm is an important open problem).  

Our first algorithm via color-coding is suitable for $b$ slightly larger than $1$ (for example, polylogarithmic).

\begin{lemma}[Algorithm for small $b$]
    \label{lem:algosmallb}
   In $\textsc{KnapsackExtendWeak}^+$ (\cref{prob:prob3}), suppose there are $m$ distinct sets among $S[0],S[1],\dots,S[L-1]$, and $|S[i]|\le b$ for all $i$.

Then, $\textsc{KnapsackExtendWeak}^+$   can be solved deterministically in $O(|U|mb^3+Lb^2 \log m)$ time.
\end{lemma}
\cref{lem:algosmallb} will be proved later in this section.
Using \cref{lem:algosmallb} as a building block, we can solve $\textsc{KnapsackExtendWeak}^+$ for even larger $b$, using another level of color-coding. This proves \cref{lem:prob3largeb}, restated below.

\algolargeb*
\begin{proof}[Proof of \cref{lem:prob3largeb} assuming \cref{lem:algosmallb}]
For a parameter $r = \max(1, b/\log m)$, use  \cref{thm:detballs} to construct in $O(|U|m\log r )$ time a coloring $h\colon U \to [r]$, such that for all $i\in \{0,1,\dots,L-1\}$ and color $c\in [r]$, $S[i] \cap h^{-1}(c) \le b'$ for some $b' =O(\log m) $.

Then, we iteratively apply the algorithm from \cref{lem:algosmallb} with size bound $b'$, to solve for each color class $U_c:=h^{-1}(c)$ ($c\in [r]$). 
More precisely, let $K$ denote the input $\textsc{KnapsackExtendWeak}^+$ instance, and starting with instance $K_1:= K$, we iterate $c\gets 1,2,\dots,r$, and do the following 
(using notations from \cref{subsec:help}):
\begin{itemize}
    \item Solve instance $(K_c)\lvert_{U_c}$ using \cref{lem:algosmallb} (with size bound $b'$), and obtain solution $Y_c$.
    \item Define instance $K_{c+1}:= (K_c)^{(U_c\gets Y_c)}$.
\end{itemize}
Finally, return $Y:= Y_{r}\circ \dots \circ Y_2\circ Y_1$. By inductively applying \cref{lem:composeweak}, we know $Y$ is a correct solution to $K$. 

We remark on the low-level implementation of the procedure described above. When we construct the new input instance for the next iteration (namely, $(K_c)^{(U_c\gets Y_c)}$), we need to prepare the new input sets $S'[\cdot ]$, based on the current input sets $S[\cdot ]$ and the current solution $z[\cdot ]$, according to \cref{eqn:nextsets}. However, each set $S[i]$ may have support size as large as $b$, so it would be to slow to copy them explicitly. In contrast, note that an instance $(K_c)\lvert_{U_c}$ only asks for sets $S[i] \cap U_c$ as input, which have much smaller size $b' = O(\log m)$. The correct implementation should be roughly as follows: at the beginning each of the $S[0],S[1],\dots,S[L-1]$ receives an integer handle, and when we need to copy the sets we actually only pass the handles. And, given the handle of a set $S[i]$ and a color class $U_c$, we can report the elements in $S[i] \cap U_c$ in $O(b') = O(\log m)$ time (because we can preprocess these intersections at the very beginning). In this way, the time complexity for preparing the input instances to \cref{lem:algosmallb} is no longer a bottleneck.

It remains to analyze the time complexity. There are $r= \max(1, b/\log m)$ applications of \cref{lem:algosmallb}, each taking $O(|U_c|m(b')^3 +L(b')^2 \log m)$ time. In total it is $O(|U|m(b')^3 + rL(b')^2 \log m) = O(|U|m\log^3 m + bL \log^2 m)$ time when $b\ge \log m$. The same $O(|U|m\log^3 m + bL \log^2 m)$ bound also holds for the $b\le \log m$ case by directly applying \cref{lem:algosmallb}. Combined with the time complexity of deterministic coloring at the beginning, the total time is $O(|U|m(\log b + \log^3 m) + bL \log^2 m)= \tilde O(|U|m+bL)$.
\end{proof}

It remains to prove \cref{lem:algosmallb}. First,  via a standard application of pairwise independent hash functions, we have the following derandomized color-coding.
\begin{lemma}
    \label{lem:detcolorcode2}
   Given $m$ sets $S_1,\dots,S_m \in [U]$ with $|S_i|\le b$, there is a deterministic algorithm in $O(Umb^3)$ time that computes $k = O(\log m)$ colorings $h_1,h_2,\dots,h_k\colon [U] \to [b^2]$, such that for every $i\in [m]$ there exists an $h_j$ that assigns distinct colors to elements of $S_i$.
\end{lemma}
\begin{proof}
   Without loss of generality, assume $U\ge b^2$, and let $\caH \subseteq \{ h\colon [U] \to [b^2]\}$ be a pairwise independent hash family samplable using $\log_2(Ub^2)+O(1)$ bits (\cref{lem:pairwise}). For each $S_i$, by a union bound over all $\binom{|S_i|}{2}$ pairs of distinct $x,y\in S_i$, we know a random $h\in \caH$ maps $S_i$ to distinct values with success probability at least $1 - \binom{b}{2}\cdot \frac{1}{b^2} \ge 1/2$. By linearity of expectation, there exists an $h\in \caH$ that makes at least half of the sets $S_1,\dots,S_m$ succeed. Such an $h$ can be found by enumerating all $O(Ub^2)$ possible seeds, each of which can be checked in $O(mb)$ time. After finding an $h$, we recurse on the remaining  unsuccessful sets. The number of unsuccessful sets gets halved in each iteration, so we terminate within $\log_2 (2m)$ iterations, and the total time complexity is $O(Umb^3)$.
\end{proof}

Now we prove \cref{lem:algosmallb}.

\begin{proof}[Proof of \cref{lem:algosmallb}]
We  apply \cref{lem:detcolorcode2} to all the $m$ distinct sets in $S[0],\dots,S[L-1]\subseteq U$, and in $O(|U|mb^3)$ time obtain $k=O(\log m)$ colorings $h_1,\dots,h_k\colon U \to [b^2]$ such that every set $S[i]$ is isolated by some $h_j$ (i.e., $S[i]$ receives distinct colors under coloring~$h_j$).

Now, based on the input  $\textsc{KnapsackExtendWeak}^+$ instance $K$, we define $k$ new instances $K^{(1)}, \dots, K^{(k)}$ as follows: for each $1\le j\le k$, let 
\[\caI^{(j)}= \{i \in  \{0,1,\dots,L-1\}:\text{ $S[i]$ is isolated by $h_j$, but not by any $h_{j'}$ ($j'<j$)}\}.\] 
Then $\{\caI^{(j)}\}_{j=1}^k$  form a partition of $\{0,1,\dots,L-1\}$.
Let instance $K^{(j)} = \big (U,\{Q_w\}_{w\in U}, S^{(j)}[\,] , q^{(j)}[\,]\big ) $ be derived from the input instance $K=\big (U,\{Q_w\}_{w\in U}, S[\,] , q[\,]\big )$ with the following modification:
\[ S^{(j)}[i] := \begin{cases} S[i] & i \in \caI^{(j)},\\ \emptyset & \text{otherwise,}\end{cases}\]
    and
\[ q^{(j)}[i] := \begin{cases} q[i] & i \in \caI^{(j)},\\ -\infty & \text{otherwise.}\end{cases}\]
Clearly, $\max(K^{(1)},\dots,K^{(k)}) = K$, so by \cref{lemma:instancemax} it suffices to solve each $K^{(j)}$ and obtain solution $Y^{(j)}$, and finally the combined solution $Y:= \max\{Y^{(1)},\dots,Y^{(k)}\}$ is a correct solution to $K$.

In each instance $K^{(j)}$, the sets $S^{(j)}[i]$ are isolated by the coloring $h_j$. So we can iteratively apply the algorithm for singletons (\cref{lem:singleton}) to solve for each color class $U_c:= h_j^{-1}(c)$ ($c\in [b^2]$), in the same fashion as in the proof of \cref{lem:prob3largeb}. More precisely, starting with instance $K_1:= K^{(j)}$, we iterate $c\gets 1,2,\dots,b^2$, and do the following (using notations from \cref{subsec:help}):
\begin{itemize}
    \item Solve instance $(K_c)\lvert_{U_c}$ using \cref{lem:singleton} in $O(L)$ time, and obtain solution $Y_c$.
    \item Define instance $K_{c+1}:= (K_c)^{(U_c\gets Y_c)}$.
\end{itemize}
Finally, return $Y^{(j)}:= Y_{c^2}\circ \dots \circ Y_2\circ Y_1$. By inductively applying \cref{lem:composeweak}, we know $Y^{(j)}$ is a correct solution to $K^{(j)}$. 

It remains to analyze the time complexity. Each of the  $k$ instances $K^{(j)}$ is solved by $b^2$ applications of the singleton algorithm in $O(L)$ time each. Hence the total time is $O(kb^2L)$. 
Combined with the deterministic coloring step at the beginning, the overall time complexity is $O(kb^2L + |U|mb^3) = O(b^2L\log m + |U|mb^3)$.
\end{proof}

\section*{Acknowledgements}
I thank Ryan Williams and Virginia Vassilevska Williams for useful discussions.

	\bibliographystyle{alphaurl} 
	\bibliography{main}
\end{document}